\documentclass[11pt,a4paper,notitlepage]{article}

\newcommand{\switchtoPLF}{\lfstyle}
\usepackage{etoolbox}
\makeatletter
\renewcommand{\@biblabel}[1]{[\switchtoPLF{}#1]}        %
\patchcmd{\@citex}%
         {\csname b@\@citeb\endcsname}%
         {\switchtoPLF\csname b@\@citeb\endcsname}{}{}  %
\makeatother

\usepackage[a4paper,margin=1.0in]{geometry}
\usepackage[british]{babel}

\usepackage{ifthen}

\usepackage[utf8]{inputenc}

\usepackage{amsmath}
\usepackage{mathtools}

\usepackage{bussproofs}

\usepackage{microtype}

\usepackage{subcaption}

\usepackage{varioref}

\usepackage{xspace}

  \usepackage{amsthm}

\DeclareMathAlphabet{\mathsfsl}{OT1}{cmss}{m}{sl}

\newcommand{\bigoh}[1]{\mathrm{O} ( #1 )}

\newcommand{\littleoh}[1]{\mathrm{o} ( #1 )}

\newcommand{\bigomega}[1]{\Omega ( #1 )}

\newcommand{\set}[1]{\{ #1 \}}

\newcommand{\setdescr}[3][\mid]{\set{ #2 #1 #3 }}

\newcommand{\setsize}[1]{\lvert#1\rvert}

\newcommand{\intersection}{\cap}

\newcommand{\union}{\cup}
\newcommand{\Union}{\bigcup}

\newcommand{\Lor}{\bigvee}

\newcommand{\limpl}{\rightarrow}

\newcommand{\olnot}[1]{\overline{#1}}

\newcommand{\eqperiod}{\enspace .}
\newcommand{\eqcomma}{\enspace ,}

\newcommand{\ie}{i.e.,\ }
 
\usepackage[colorlinks=true,citecolor=blue]{hyperref}

\newtheorem{standardlocalcounter}{Dummy}[section]

\theoremstyle{plain}

\newtheorem{theorem}[standardlocalcounter]{Theorem}
\newtheorem{lemma}[standardlocalcounter]{Lemma}
\newtheorem{proposition}[standardlocalcounter]{Proposition}
\newtheorem{corollary}[standardlocalcounter]{Corollary}
\newtheorem{observation}[standardlocalcounter]{Observation}

\theoremstyle{definition}

\newtheorem{definition}[standardlocalcounter]{Definition}

\theoremstyle{remark}

\makeatletter
\preto\wrong@fontshape{\expandafter\gdef\csname \curr@fontshape/sub\endcsname{}}
\makeatother
\usepackage[charter,cal=cmcal]{mathdesign}
\usepackage[scaled=.96,sups,osf]{XCharter}
\usepackage[scale=0.95]{cantarell}
\SetMathAlphabet{\mathsf}{normal}{OT1}{fca}{m}{n}
\usepackage[scaled=1.02]{inconsolata}
\SetMathAlphabet{\mathtt}{normal}{OT1}{zi4}{m}{n}

\usepackage[dvipsnames]{xcolor}
\usepackage{bussproofs}
\newcommand{\ClauseStrut}{\strut}
\renewcommand{\Axiom}[1]{\AxiomC{\ClauseStrut$#1$}}
\renewcommand{\UnaryInf}[1]{\UnaryInfC{\ClauseStrut$#1$}}
\renewcommand{\BinaryInf}[1]{\BinaryInfC{\ClauseStrut$#1$}}

\usepackage{algorithmicx}
\usepackage{algpseudocode}
\usepackage{tikz}
\usetikzlibrary{trees,positioning}
\usepackage{tikz-cd}
\usepackage{microtype}

\usepackage[normalem]{ulem}
\usepackage{stmaryrd}
\newcommand{\ib}[1]{\llbracket{#1}\rrbracket}
\newcommand{\dl}{\mathrm{dl}}
\newcommand{\res}{\mathrm{Res}}

\newcommand{\closure}{\mathrm{Cl_i}}

\newcommand{\restrict}[2]{{#1}{\upharpoonright_{#2}}}
\newcommand{\vars}[1]{\mathrm{vars(#1)}}
\newcommand{\ONE}{{\lfstyle1}}
\newcommand{\WW}{\mathcal{W}}
\newcommand{\XX}{\mathcal{X}}
\newcounter{refutations}
\newcommand{\newrefutation}{%
  \refstepcounter{refutations}%
  \emph{refutation~\arabic{refutations}}%
  \label{ref:\arabic{refutations}}%
}

\newif\iffullversion
\fullversiontrue

\title{Limits of CDCL Learning via Merge Resolution}
\author{Marc Vinyals \and Chunxiao Li \and Noah Fleming \and Antonina Kolokolova \and Vijay Ganesh}

\begin{document}

\maketitle

\begin{abstract}
In their seminal work, Atserias et al. and independently Pipatsrisawat and Darwiche in 2009 showed that CDCL solvers can simulate resolution proofs with polynomial overhead. However, previous work does not address the tightness of the simulation, i.e., the question of how large this overhead needs to be. In this paper, we address this question by focusing on an important property of proofs generated by CDCL solvers that employ {\it standard learning schemes}, namely that the derivation of a learned clause has at least one inference where a literal appears in both premises (aka, a merge literal). Specifically, we show that proofs of this kind can simulate resolution proofs with at most a linear overhead, but there also exist formulas where such overhead is necessary or, more precisely, that there exist formulas with resolution proofs of linear length that require quadratic CDCL proofs.
 \end{abstract}

\section{Introduction}

Over the last two decades, CDCL SAT solvers have had a dramatic impact on many areas of software engineering~\cite{cadar2008exe}, security~\cite{dolby2007security,xie2005security}, and AI~\cite{blum1997fast}. This is due to their ability to solve very large real-world formulas that contain upwards of millions of variables and clauses~\cite{HandbookCDCL}. Both theorists and practitioners have expended considerable effort in understanding the CDCL algorithm and the reasons for its unreasonable effectiveness in the context of practical applications. While considerable progress has been made, many questions remain unanswered.

Perhaps the most successful set of tools for understanding the CDCL algorithm come from proof complexity, and a highly influential result is the one that shows that idealized models of CDCL can polynomially simulate the resolution proof system, proved independently by Atserias, Fichte, and Thurley~\cite{AFT11ClauseLearning}, and Pipatsrisawat and Darwiche~\cite{PD11OnThePower}, building on initial results by Beame et al.~\cite{BKS04TowardsUnderstanding} and Hertel et al.~\cite{HBPV08ClauseLearning}. (See also a recent alternative proof by Beyersdorff and Böhm~\cite{BB21QBF}.)
Such simulation results are very useful because they reassure us that whenever a formula has a short resolution proof then CDCL with the right choice of heuristics can reproduce it.

Recent models make assumptions that are closer to real solvers, but pay
for that with a polynomial overhead in the simulation. A series of papers have focused on understanding which of the assumptions are needed for these simulations to hold, often using and/or introducing refinements of resolution along the way.
For instance, the question of whether restarts are needed, while still
open, has been investigated at length, and the pool
resolution~\cite{VanGelder05PoolResolution} and
RTL~\cite{BHJ08ResolutionTrees} proof systems were devised to capture
proofs produced by CDCL solvers that do not restart. 
The importance of decision heuristics has also been explored recently,
with results showing that neither static~\cite{MPR20CDCL} nor
VSIDS-like~\cite{Vinyals20HardExamples} ordering of variables are enough to simulate
resolution in full generality (unless VSIDS scores are periodically erased~\cite{li2020towards}). In the case of static ordering, the (semi-)ordered
resolution proof system~\cite{MPR20CDCL} was used to reason about such variants of CDCL solvers. 

But even if we stay within the idealized model, it is not clear how efficient
CDCL is in simulating resolution. The analysis of Pipatsrisawat and
Darwiche gives an $\bigoh{n^4}$ overhead---that is, if a formula over
$n$ variables has a resolution refutation of length $L$, then a CDCL
proof with no more than $\bigoh{n^4L}$ steps exists. Beyersdorff and Böhm~\cite{BB21QBF} improved the overhead to $\bigoh{n^3}$, but we do not
know what the optimal is. Furthermore, to the best of our knowledge, prior
to our paper, we did not even know if the overhead can be avoided altogether.

\subsection{Learning Schemes in CDCL and Connection with Merges}

A common feature of CDCL solvers is the use of 1-empowering learning schemes~\cite{PD08Learning,AFT11ClauseLearning}: that is, they only learn clauses which enable unit propagations that were not possible before. An example of 1-empowering learning scheme is the popular \ONE UIP learning scheme~\cite{MS99Grasp}.  To model this behavior we build upon a connection between 1-empowerment, and merges~\cite{Andrews68Resolution}, i.e., resolution steps involving clauses with shared literals.
Nearly every CDCL solver nowadays uses the
First Unique Implication Point (\ONE UIP) learning scheme, where
conflict analysis  starts with a clause falsified by the
current state of the solver and sequentially resolves it with clauses
responsible for unit propagations leading to the conflict,
until the clause becomes \emph{asserting}, i.e., unit immediately upon backjumping.

Descriptions of early implementations of CDCL
solvers~\cite{MS99Grasp,MMZZM01Engineering} already remark on the
importance of learning an asserting clause, since that nudges the
solver towards another part of the search space, and consequently early
alternative learning schemes explored learning many kinds of asserting
clauses.
First observe that conflict analysis can be extended to produce other
asserting clauses that appear after the \ONE UIP during conflict
analysis such as intermediate UIPs and the last
UIP~\cite{BS97UsingCSP}. The early solver GRASP can even learn
multiple UIP clauses from a single conflict.  While there is empirical
evidence that it is often best to stop conflict analysis at the \ONE
UIP~\cite{ZMMM01EfficientConflict}, recent work has identified
conditions where it is advantageous to continue past
it~\cite{FB20ClauseSize} (see also the discussion of learning schemes therein).

Ryan~\cite[\S 2.5]{Ryan04Thesis} also observed empirically that
clause quality is negatively correlated with the length of the
conflict analysis derivation and considered the opposite approach, that
is, learning clauses that appear before the \ONE UIP during conflict
analysis in addition to the \ONE UIP. This approach is claimed to be useful for
some empirical benchmarks but, like any scheme that learns multiple
clauses, slows down Boolean constraint propagation (BCP) in comparison
to a scheme that learns just the \ONE UIP.

Later works provide a more theoretically oriented approach to
understanding the strength of \ONE UIP and to learning clauses that
appear before the \ONE UIP~\cite{DHN07Towards,PD08Learning}.
In particular, and highly relevant for our discussion, Pipatsrisawat
and Darwiche identified 1-empowerment as a fundamental property of asserting clauses.
Furthermore they identified a connection between 1-empowering clauses
and merges, and used the simplicity of checking for merges as an approximation for
1-empowerment.

An orthogonal approach is to extend the \ONE UIP derivation
by resolving it with clauses other than those that would usually be
used during conflict analysis~\cite{ABHJS08Conflict}. A prominent
example is clause minimization~\cite{SB09Minimizing}, where literals
are eliminated from the \ONE UIP clause by resolving it with the
appropriate input clauses, independently of their role in
the conflict, so the resultant clause that is actually learned is a shorter and
therefore stronger version of the \ONE UIP.

Furthermore, a relation between merges and unit-resolution completeness has also
been observed in the context of knowledge
compilation~\cite{delVal94Tractable}. Finally, 
 the amount of merges directly
inferable from a formula (\ie in a single resolution step) has been
proposed,  under the name of mergeability, as a measure to help explain
the hardness of a formula based on both controlled experiments as well as analysis of real-world instances~\cite{ZMWLCG18Effect}.

To summarize, merges are relevant in the context of CDCL learning schemes for the following reason: all practical CDCL learning schemes either produce a
1-empowering clause or extend one, and since 1-empowering clauses
always contain a merge in its derivation, we have that all practical learning
schemes produce a clause that contains a merge in its derivation,
which is exactly the property imposed by the proof systems we introduce below.

\subsection{Our Contributions}

As mentioned earlier, we build upon a connection between 1-empowerment and merges~\cite{PD08Learning,AFT11ClauseLearning},
and
introduce a proof system RMA (for ``resolution with merge ancestors'') which includes CDCL with an arbitrary 1-empowering learning scheme.
The ``merge ancestors'' in the name of this system comes from the fact that for any 1-empowering clause, at least one step in its resolution derivation must resolve two clauses that share a common literal: a \emph{merge} step in the sense of  \cite{Andrews68Resolution}.   Clause minimization procedures, as long as they are applied on top of 1-empowering clauses, are also modelled by RMA.

We prove that, on the one hand, RMA is able to
simulate resolution only with a linear overhead.
On the other hand, we show a quadratic separation between resolution
and RMA, that is there exist formulas with resolution proofs of linear
length that require %
RMA proofs of quadratic length. That is, we show that CDCL may be polynomially worse than resolution because of the properties of a standard learning scheme,  but that the blow-up due to these properties is not more than linear.

We also consider weaker proof systems, all of which contain \ONE UIP (and do so with finer granularity), but not necessarily other asserting learning schemes. A technical point of interest is that we work with proof systems that are provably not closed under restrictions,
which is unusual in proof complexity. This fact forces
our proof to exploit syntactic properties of
the proof system, as opposed to relying on more convenient semantic
properties.

\section{Preliminaries}
\label{sec:prelims}

A literal is either a variable $x^1=x$ or its negation
$x^0=\olnot{x}$. A clause is a disjunction of literals, and a CNF
formula is a conjunction of clauses.
The support of a clause or $\vars{C}$ is the set of variables it contains.
A resolution derivation
from a formula $F$ is a sequence of clauses $\eta=C_1,\ldots,C_L$ such that
$C_i$ is either an axiom in $F$ or it is the conclusion of applying
the resolution rule
\[\res(A \lor x, B \lor \olnot{x}) = A \lor B\]
on two premises $C_{j}$, $C_{k}$ with $j,k<i$.
The variable $x$ that appears with opposite signs in the premises of a
resolution inference is called the pivot. If furthermore there is a
\emph{literal} common to $A$ and $B$ the resolvent is called a merge.
If instead of being the result of a syntactic inference we allow $C_i$ to be any clause semantically implied by
$C_{j}$ and $C_{k}$, even if $C_j$ and $C_k$ might not be resolvable, then we say $\eta$ is a semantic resolution
derivation. A derivation is a refutation if its last clause is the empty clause $\bot$. We denote $\eta[a,b]=\setdescr{C_i\in\eta}{i\in[a,b]}$. 

We assume that every clause in a derivation is annotated with the
premises it is obtained from, which allows us to treat the proof as a
DAG where vertices are clauses and edges point from premises to
conclusions. When this DAG is a tree we call a derivation tree-like,
and when it is a centipede (\ie a maximally unbalanced
tree) we call it input.

A derivation is unit if in every inference at least one of the
premises is a unit clause consisting of a single literal. Since
neither input nor unit resolution are complete proof systems, we write
$F \vdash_i C$ (respectively $F \vdash_1 C$) to indicate that there
exists an input (resp. unit) resolution derivation of $C$ from $F$.

A clause $C$ syntactically
depends on an axiom $A$ with respect to a
derivation $\eta$ if there is a path from $A$ to $C$ in the DAG
representation of $\eta$. This does not imply that $A$ is required to derive $C$, since a different derivation might not use $A$. %

A restriction to variables is a mapping
$\rho\colon X \to X \union\set{0,1}$, successively extended to
literals, clauses, formulas, and refutations, simplifying where
needed. We write $\rho(x)=*$ as a shorthand for $\rho(x)=x$. It is
well-known that if $\eta$ is a resolution derivation from $F$ and
$\rho$ is a restriction, then $\restrict{\eta}{\rho}$ is a semantic
resolution derivation from $\restrict{F}{\rho}$.

It is convenient to leave satisfied clauses in place in a derivation
that is the result of applying a restriction to another derivation so
that we can use the same indices to refer to both derivations. To do
that we use the symbol $1$ and treat it as a clause that always
evaluates to true, is not supported on any set, does not depend on any
clause, and cannot be syntactically resolved with any clause.

A semantic derivation can be turned into a syntactic derivation by
ignoring unnecessary clauses. Formally,
if $\eta$ is a semantic resolution derivation, we define its syntactic equivalent $s(\eta)$ as
the following syntactic resolution derivation. Let $C \in \eta$ and
let $A$ and $B$ be the parents of $C$. If $s(A) \vDash C$ we set
$s(C) = s(A)$, analogously with $s(B)$. Otherwise we set
$s(C) = \res(s(A), s(B))$.
It is not hard to see that for each $C_i \in \eta$, $s(C_i) \vDash C_i$.

\iftrue
\subsection{CDCL}

We need to define a few concepts from CDCL proofs. An in-depth
treatment can be found in the Handbook of
Satisfiability~\cite{BN21Proof}. Fix a CNF $F$, also known as clause database. A trail $\tau$ is a sequence of tuples $(x_{j_i}=b,C_i)$ where $C_i$ is either a clause in
$F$ or the special symbol $d$ representing a decision. We denote by
$\alpha_{< i}$ the assignment $\setdescr{x_{j_i}=b}{i'<i}$, and
we denote by $\dl(i)=\dl(i-1)+\ib{C_i=d}$ the decision level at
position $i$, that is the number of decisions up to $i$.
We mark the position of the last decision in a trail by $i^*$.

A trail is valid if for every position $i$ that is not a decision we
have that $\restrict{C_i}{\alpha_{<i}}=x_{j_i}^b$ and for every
decision $i$ we have that for every clause $C\in F$ such that
$\restrict{C}{\alpha_{<i}} = x^b$, the literal $x^b$ appears in the
trail before $i$. In particular, for every position $i'<i$ with
$\dl(i')<\dl(i)$ we have $\restrict{C_i}{\alpha_{<i'}} \neq x_{j_i}^b$.

A clause $C$ is asserting if it is unit at the last decision in the
trail, that is $\restrict{C}{\alpha_{<i^*}}=x^b$.
It is 1-empowering if $C$ is implied by $F$ and can lead to new unit
propagations after being added to $F$, that is if there exists a literal
$\ell\in C$ such that for some $A\in\set{\bot,\ell}$, it holds that
$F \land\olnot{C\setminus\ell} \nvdash_1 A$. If a clause is not 1-empowering then we say it is absorbed by $F$.

Given a clause $D_{\setsize{\tau}}$ falsified by a trail $\tau$, the
conflict derivation is an input derivation
$D_{\setsize{\tau}},\ldots,D_{k+1},D_k$ where $D_{i-1}= \res(D_i,C_i)$ if
$x_{j_i} \in D_i$, and $D_{i-1} = D_i$ otherwise. The first (\ie with
the largest index) asserting clause in the derivation is called the
\ONE UIP. Note that $D_{i^*}$ is always asserting (because $D_i$ is falsified
by $\alpha_{\leq i}$ for $i^* \leq i \leq \setsize{\tau}$ and $D_{i^*}$
is not falsified by $\alpha_{< i^*}$), therefore we can assume that the
\ONE UIP always has index at least $i^*$.

We call a sequence of input derivations \emph{input-structured} if the
last clause of each derivation can be used as an axiom in successive
derivations. The last clause of each but the last derivation is called
a lemma.
A CDCL derivation is an input-structured sequence of conflict
derivations, where learned clauses are lemmas. This definition is
similar to that of Resolution Trees
with Input Lemmas~\cite{BHJ08ResolutionTrees}, with the difference
that the sequence only needs to be ordered, without imposing any further
tree-structure on the global proof.

The following Lemmas highlight the practical relevance of merges by relating them to \ONE UIP, asserting, and 1-empowering clauses.

\begin{lemma}[{\cite[Proposition 2]{PD08Learning}}]
  \label{lem:asserting-empowering}
  If a clause is asserting, then it is 1-empowering.%
  \footnote{The original result does not prove 1-consistency, but the
    proof is analogous.}
\end{lemma}

\begin{lemma}[{\cite[Lemma 8]{AFT11ClauseLearning}}]
  \label{lem:empowering-merge}
  If $A \lor x$ and $B \lor \olnot{x}$ are absorbed but $A \lor B$ is
  1-empowering, then $A \lor B$ is a merge.
  In particular, if a clause is 1-empowering, then it contains a merge in its
  derivation.
\end{lemma}

\begin{lemma}
  \label{lem:1uip-is-merge}
  The \ONE UIP clause is a merge.
\end{lemma}

\begin{proof}
  Let $D_{j} = \res(C_{j+1},D_{j+1})$ be the \ONE UIP.
  On the one hand, since every clause in the trail contains at least
  two literals at the same decision level it appears in, $C_{j+1}$
  contains two literals at the last decision level.
  On the other hand, any clause that is not in the trail also contains
  two literals at the last decision level, and in particular
  $D_{\setsize{\tau}}$. Since $\setsize{D_{i+1} \setminus D_i} \leq 1$
  and $D_{j+1}$ is not asserting, it also contains two literals at the
  last decision level.

  We accounted for 4 literals at the last decision level present in
  the premises of $D_j$, of which 2 are not present in the conclusion
  because they are the pivots. In order for $D_j$ to contain only one
  literal at the last decision level, the remaining two literals must
  be equal.
\end{proof}

\fi

\section{Proof Systems}
\label{sec:models}

We define our proof systems in terms of the \emph{input-structured} framework. Every resolution proof can be thought of as being input-structured if we consider it as a sequence of unit-length input resolutions and every clause as a lemma; it is when we impose restrictions on which clauses are permitted as lemmas that we obtain different proof systems. The diagram in Figure~\ref{fig:ps} can help keeping track of the proof systems.

\iffullversion\newcommand{\refprefixeccc}{-x}\else\newcommand{\refprefixeccc}{}\fi

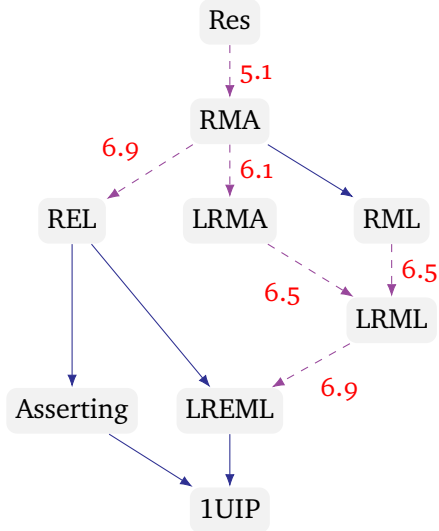
\begin{figure}[ht]
    \tikzset{ps/.style={fill=black!5,rounded corners},execute at end node=\vphantom{fg}}
    \tikzset{algo/.style={fill=black!5,rounded corners},execute at end node=\vphantom{fg}}
    \tikzset{sim/.style={-Latex,draw=Blue}}
    \tikzset{sep/.style={-Latex,dashed,draw=Purple}}
    \begin{center}
      \begin{tikzpicture}[auto]
        \node[ps] (res) {Res};
        \node[ps,below = .7cm of res] (rma) {RMA};
        \node[ps,below = .7cm of rma] (srrma) {LRMA};
        \node[ps,left = of srrma] (rel) {REL};
        \node[ps,right = of srrma] (rml) {RML};
        \node[ps,below = .7cm of rml] (srrml) {LRML};
        \node[ps, below = 1.9cm of srrma] (greedy) {LREML};
        \node[algo, below = .7cm of greedy] (fuip) {\ONE{}UIP};
        \node[algo, below = 1.9cm of rel] (ass) {Asserting};
        \draw[sep] (res) to node {\ref{th:separation}} (rma);
        \draw[sep] (rma) to node[swap] {\ref{prop:srrml-rel\refprefixeccc}} (rel);
        \draw[sep] (rma) to node {\ref{prop:rma-srrma\refprefixeccc}} (srrma);
        \draw[sim] (rma) to (rml);
        \draw[sim] (rel) to (ass);
        \draw[sim] (rel) to (greedy);
        \draw[sep] (rml) to node {\ref{prop:rml-srrma-srrml\refprefixeccc}} (srrml);
        \draw[sep] (srrma) to node[swap] {\ref{prop:rml-srrma-srrml\refprefixeccc}} (srrml);
        \draw[sep] (srrml) to node {\ref{prop:srrml-rel\refprefixeccc}} (greedy);
        \draw[sim] (greedy) to (fuip);
        \draw[sim] (ass) to (fuip);
      \end{tikzpicture}
    \end{center}
    \caption[Relations between proof systems]{Relations between proof systems. A solid arrow
    \raisebox{0pt}[0pt][0pt]{\tikz[baseline]{
      \node[anchor=base] (a) {$A$};
      \node[right = 1.5em of a] (b) {$B$};
      \draw[sim] (a) to (b);
    }} indicates that $A$ simulates $B$ with no overhead.
    A dashed arrow
    \raisebox{0pt}[0pt][0pt]{\tikz[baseline]{
      \node[anchor=base] (a) {$A$};
      \node[right = 1.5em of a] (b) {$B$};
      \draw[sep] (a) to (b);
    }} indicates that $A$ simulates $B$ with no overhead, but $B$ requires linear overhead to simulate $A$. Statements proving separations are referenced.
  }
  \label{fig:ps}
\end{figure}

Andrews' definition of merge resolution~\cite{Andrews68Resolution} considers tree-like
proofs with the additional restriction that in every inference at least
one premise is an axiom or a merge. He also observes that such derivations can be made
input-structured.

\begin{observation}[\cite{Andrews68Resolution}]
  \label{obs:input-structure}
  A tree-like merge resolution derivation can be decomposed into an input-structured
  sequence where all the lemmas are merges.
\end{observation}

This observation is key when working with such derivations, as is apparent in Sections~\ref{sec:simulation} and~\ref{sec:extra},
to the point that we use as an alternative way to
define merge resolution.

Andrews' main result is that the merge restriction does not affect tree-like resolution.
\begin{lemma}[{\cite[Lemma 5]{Andrews68Resolution}}]
  \label{lem:tree-to-input}
  If there is a tree-like resolution derivation of $C$ of length $L$
  where at most the root is a merge, then there is an input resolution
  derivation of some $C' \subseteq C$ of length at most $L$.
\end{lemma}

\begin{theorem}[{\cite[Theorem 1]{Andrews68Resolution}}]
  \label{th:tree-to-merge}
  If there is a tree-like resolution derivation of $C$ of length $L$,
  then there is a tree-like merge resolution derivation of some
  $C' \subseteq C$ of length at most $L$.
\end{theorem}

If we lift the tree-like restriction from the input-structured view of merge resolution proofs we obtain a proof
system between tree- and DAG-like resolution where clauses can be reused (\ie have
outdegree larger than $1$) if and only if they are merges or, in other
words, lemmas in the input-structured decomposition.
We call this proof system
Resolution with Merge Lemmas
and refer to it with the acronym RML.
\begin{definition}
  \label{def:rml}
A RML derivation is an input-structured sequence of unit resolution derivations where all lemmas are merges.
\end{definition}

CDCL refutations produced by solvers that use the \ONE{}UIP learning
scheme are in RML form, as a consequence of
Lemma~\ref{lem:1uip-is-merge}.
We can also generalize RML to allow reusing clauses that contain a
merge anywhere in their derivation. We call this proof system
Resolution with Merge Ancestors, or RMA for short.
\begin{definition}
  \label{def:rma}
  A RMA derivation is an input-structured sequence of unit resolution derivations where all derivations but the last contain a merge.
\end{definition}

Note that by Lemma~\ref{lem:tree-to-input} it does not matter if we
require the sequence of derivations of an RMA derivation to be input derivations or if we
allow general trees.
In fact, our lower bound results hold for a more general proof system where we only ask
that every
clause with outdegree larger than $1$ has an ancestor that is a merge.
Such proof system does not have a simple input structure, but can rather
be thought of as a sequence of tree-like resolution derivations whose
roots are merges, followed by a standard resolution derivation using
the roots of the previous derivations as axioms.

To make the connection back to CDCL, we can define a proof system called Resolution with Empowering Lemmas that captures CDCL refutations produced by solvers that use any asserting learning scheme or 1-empowering learning scheme.
\begin{definition}
  \label{def:rel}
  Let $C_1,\ldots,C_{L-1}$ be the lemmas of an input-structured sequence of unit derivations. The sequence is a Resolution with Empowering Lemmas (REL) derivation of a formula $F$ if $C_i$ is 1-empowering with respect to $F \union \set{C_j : j<i}$ for all $i\in[1,L-1]$.
\end{definition}

It follows from Lemmas~\ref{lem:asserting-empowering}
and~\ref{lem:empowering-merge}
that such refutations are in RMA form.

\begin{observation}
  \label{def:rel-rma}
  A REL derivation is a RMA derivation.
\end{observation}

It might seem more natural to work with the REL proof system rather than its merge-based counterparts, since REL is defined exactly through the 1-empowering property. However, while the merge property is easy to check because it is local to the derivation at hand, we can only determine if a clause is 1-empowering by looking at the full history of the derivation, in particular what the previous lemmas are. This makes REL too cumbersome to analyse.
Furthermore, CDCL refutations produced apply a clause minimization scheme on top of an asserting clause might not be in REL form, but they are still in RMA form.

A further property of input derivations produced by a CDCL solver is that once a variable is resolved, it does not appear later in the derivation.

\begin{definition}
  \label{def:srrml-srrma}
  A resolution derivation $\eta$ is strongly regular if for every resolution step $i$, the pivot variable $x_i$ is not part of the support of any clause $C_i\in\eta[i,L]$.
  A sequence of derivations is locally regular if every derivation in the sequence is strongly regular.
  A LRML derivation (resp. LRMA) is a locally regular RML derivation (resp. RMA).
\end{definition}

Finally we can consider derivations that have empowering, merge lemmas and are locally regular. These still include \ONE{}UIP proofs.

\begin{definition}
  \label{def:greedy}
  A LREML derivation is a derivation that is both LRML and REL.
\end{definition}

It follows from the simulation of resolution by CDCL~\cite{PD11OnThePower,AFT11ClauseLearning} that all (DAG-like) proof systems we defined
polynomially simulate standard resolution.
In Section~\ref{sec:simulation} we make this simulation more precise and
prove that the simulation overhead can be made linear, and in Section~\ref{sec:separation} that the simulation is optimal because there exist formulas that have resolution refutations of linear length but require RMA
refutations of quadratic length.
\section{Simulation}
\label{sec:simulation}

As an auxiliary tool to simulate resolution in RML we define the
input-resolution closure of a set $D$, denoted
$\closure(D) = \setdescr{C}{\exists C'\subseteq C,\,D\vdash_i C'}$, as
the set of clauses derivable from $D$ via input resolution plus
weakening. It is well-known that, since input resolution derivations
can be assumed to be strongly regular without loss of generality, we can also
assume them to be at most linear in the number of variables.

\begin{observation}
  \label{obs:input-short}
  If $D$ is a CNF formula over $n$ variables and $C \in \closure(D)$ then
  there is a strongly regular input resolution derivation of some $C' \subseteq C$
  from $D$ of length at most $n$.
\end{observation}

Combining Theorem~\ref{th:tree-to-merge} with the idea that in order
to simulate a resolution derivation we do not need to generate each
clause, but only do enough work so that in the following steps we can
pretend that we had derived
it~\cite{PD11OnThePower,AFT11ClauseLearning}, we can prove that merge
resolution simulates resolution with at most a multiplicative linear
overhead in the number of variables.

\begin{theorem}
  \label{th:merge-sim-res}
  If $F$ is a CNF formula over $n$ variables that has a resolution
  refutation of length $L$ then it has a RML refutation of length
  $\bigoh{nL}$.
\end{theorem}

\begin{proof}
  Let $\pi=(C_1,\ldots,C_L)$ be a resolution refutation. We construct a sequence of sets $D_0,\ldots,D_L$ with the following properties.

  \begin{enumerate}
  \item $D_t \setminus F$ is the set of lemmas in a RML derivation of length at most $(2n+1)t$.
  \item $\pi[1,t] \subseteq \closure(D_t)$.
  \end{enumerate}

  This is enough to prove the theorem: since $\bot \in D_t$ we can
  obtain $\bot$ from $D_t$ in length $n$, so the total length of the
  refutation is $(2n+1)L+n$.

  We build the sets by induction, starting with $D_0 = F$.
  Assume we have built $D_t$ and let $C = C_{t+1} = \res(A,B)$ with
  $A,B \in \pi[1,t]$. If $C \in \closure(D_t)$ we set $D_{t+1}=D_t$
  and we are done. Otherwise, by induction we have
  $A,B \in \closure(D_t)$, therefore by
  Observation~\ref{obs:input-short} there are input resolution
  derivations of $A' \subseteq A$ and $B' \subseteq B$ of length at
  most $n$. Since neither $A' \vDash C$ nor $B' \vDash C$, $A'$ and
  $B'$ can be resolved and therefore there is a tree-like derivation
  $\eta$ of $C' \subseteq C$ from $D_t$ of length at most $2n+1$. By
  Theorem~\ref{th:tree-to-merge} there is a tree-like merge resolution derivation
  $\eta'$ of $C''\subseteq C$ from $D_t$ of length at most $2n+1$. By
  Observation~\ref{obs:input-structure} the derivation $\eta'$ can be decomposed into a sequence of input derivations of total length at most $2n+1$.
  Let $E$ be the lemmas in that sequence and set
  $D_{t+1} = D_t \union E$. We have that
  $C \in \closure(F \union E) \subseteq \closure(D_{t+1})$, and that we can
  obtain $E$ from $D_t$ in at most $2n+1$ steps. Thus $D_{t+1}$ has
  all the required properties.
\end{proof}

We can be a bit more precise with the description of the simulation if
we look at the structure of $\eta$ before applying
Theorem~\ref{th:tree-to-merge}. Let $A_M$ and $B_M$ be the last merges
in the input derivation of $A'$ and $B'$ respectively, and let
$E=\set{A_M,B_M}$.

Now consider the fragment of the input derivation of $A'$ from $A_M$
to $A'$, analogously with $B'$. We have a tree-like derivation of $C'$
where at most the root is a merge, therefore we can apply
Lemma~\ref{lem:tree-to-input} directly instead of
Theorem~\ref{th:tree-to-merge} and obtain an input resolution
derivation of $C'' \subseteq C$ from $E \union F$.

If we also make sure that the input derivations of $A'$ and $B'$ are strongly regular, we have that LRML can also simulate resolution with the same $\bigoh{n}$ overhead as RML.

An analogous result can be obtained for LREML from the following lemma.

\begin{lemma}[\cite{PD11OnThePower}]
\label{lem:absorbed-unit} If $F$ absorbs $A \lor x$ and $B \lor \olnot{x}$, then $F \vdash_i C' \subseteq A \lor B$.
\end{lemma}

\begin{corollary}
  \label{cor:emp-sim-res}
  If $F$ is a CNF formula over $n$ variables that has a resolution
  refutation of length $L$ then it has a LREML refutation of length
  $\bigoh{nL}$.
\end{corollary}

\begin{proof}
  The proof follows the general structure of
  Theorem~\ref{th:merge-sim-res}, except that we use a sequence of
  steps $D_{t}^j$ in order to construct $D_{t}$. Our induction
  hypothesis is that $D_{t}^{j}$ can be derived from $D_t$ in $p$
  inference steps in LREML, and that $A'$ and $B'$ can be derived from $D_{t}^{j}$ in $q$ steps,
  with $p+q \leq 2n$.

  The base case $D_{t}^0=D_{t}$ is trivial.

  For the inductive case, assume that the input derivations leading to
  $A'$ and $B'$ are strongly regular without loss of generality.  By
  Lemma~\ref{lem:absorbed-unit} either $A'$ or $B'$ is 1-empowering,
  say $A'$. Let $C$ be the first 1-empowering clause in the derivation
  of $A'$. By Lemma~\ref{lem:empowering-merge} $C$ is a merge,
  therefore we can take $D_{t}^{j+1}=D_{t}^{j}\union\set{C}$.
\end{proof}

\section{Separation}
\label{sec:separation}

We prove the following separation between standard resolution and RMA.

\begin{theorem}
  \label{th:separation}
  There exists a family of formulas $F_n$ over $\bigoh{n\log n}$ variables and
  $\bigoh{n\log n}$ clauses that have resolution refutations of length
  $\bigoh{n\log n}$ but every RMA refutation requires length
  $\bigomega{n^2 \log n}$.
\end{theorem}

\subsection{Formula}

Let $\ell,m,n$ be positive integers. We have variables $x_i$ for
$i\in[m\ell-1]$ and $w_{j,k}$ for $j\in[\ell]$ and $k\in[n]$. For convenience we define
$x_0=1$ and $x_{m\ell}=0$, which are not variables.
Let
$X=\setdescr{x_i}{i\in[m\ell-1]}$, $W_j=\setdescr{w_{j,k}}{k\in[n]}$
and $W=\Union_{j\in[\ell]} W_j$. For each $j\in[\ell]$ we build the
following gadget:
\begin{align}
  &w_{j,k}=w_{j,k+1} && \text{for $k\in[n-1]$}
\end{align}
Each equality is expanded into the two clauses
$B_{j,k,1}=w_{j,k} \lor \olnot{w_{j,k+1}}$ and $B_{j,k,0}=\olnot{w_{j,k}} \lor w_{j,k+1}$, and we collectively call them $\WW = \setdescr{B_{j,k,b}}{j\in[\ell],k\in[n-1],b\in\set{0,1}}$. Observe that the $j$-th gadget implies $w_{j,1}=w_{j,n}$.
Additionally we build the following gadget:
\begin{align}
  &(w_{1,1} = w_{1,n}) \limpl x_1\\
  &(w_{\hat\imath,1} = w_{\hat\imath,n}) \limpl ( x_{i-1} \limpl x_{i} ) && \text{for $i\in[2,m\ell-1]$}\\
  &(w_{\ell,1} = w_{\ell,n}) \limpl \olnot{x_{m\ell-1}}
\end{align}
where $\hat\imath\in[\ell]$ denotes the canonical form of $i \pmod {\ell}$. Each
constraint is expanded into the two clauses
$A_{i,1}=w_{\hat\imath,1} \lor w_{\hat\imath,n} \lor \olnot{x_{i-1}} \lor x_i$ and
$A_{i,0}=\olnot{w_{\hat\imath,1}} \lor \olnot{w_{\hat\imath,n}} \lor \olnot{x_{i-1}} \lor
x_i$, and we collectively call them $\XX = \setdescr{A_{i,b}}{i\in[m\ell],b\in\set{0,1}}$. The resulting
formula is called $F_{\ell,m,n}$.

\subsection{Upper Bound}
\label{sec:ub}

It is not hard to see that there is a resolution refutation of
$F_{\ell,m,n}$ of length $\bigoh{\ell\cdot(m+n)}$. Indeed, we first
derive the two clauses representing $w_{j,1} = w_{j,n}$ for each
$j \in [\ell]$, which requires $\bigoh{n\ell}$ steps:
\begin{equation}
  \label{eq:ref-w1n}
  \Axiom{w_{j,1} \lor \olnot{w_{j,2}}}
  \Axiom{w_{j,2} \lor \olnot{w_{j,3}}}
  \BinaryInf{w_{j,1} \lor \olnot{w_{j,3}}}
  \UnaryInf{\vdots}
  \UnaryInf{w_{j,1} \lor \olnot{w_{j,n-1}}}
  \Axiom{w_{j,n-1} \lor \olnot{w_{j,n}}}
  \BinaryInf{w_{j,1} \lor \olnot{w_{j,n}}}
  \DisplayProof
\end{equation}
Then we resolve
each of the $\XX$ axioms with one of these clauses, appropriately
chosen so that we obtain pairs of clauses of the form
$w_{\hat\imath}^b \lor \olnot{x_{i-1}} \lor x_i$ for $i\in[m\ell]$,
and resolve each pair to obtain the chain of implications $x_1, \ldots,
x_i \limpl x_{i+1}, \ldots, \olnot{x_{n\ell-1}}$ in $\bigoh{m\ell}$ steps.
\begin{equation}
  \label{eq:ref-xi}
  \Axiom{w_{\hat\imath,1} \lor \olnot{w_{\hat\imath,n}}}
  \Axiom{w_{\hat\imath,1} \lor w_{\hat\imath,n} \lor \olnot{x_{i-1}} \lor x_i}
  \BinaryInf{w_{\hat\imath,1} \lor \olnot{x_{i-1}} \lor x_i}
  \Axiom{\olnot{w_{\hat\imath,1}} \lor w_{\hat\imath,n}}
  \Axiom{\olnot{w_{\hat\imath,1}} \lor \olnot{w_{\hat\imath,n}} \lor \olnot{x_{i-1}} \lor x_i}
  \BinaryInf{\olnot{w_{\hat\imath,1}} \lor \olnot{x_{i-1}} \lor x_i}
  \BinaryInf{\olnot{x_{i-1}} \lor x_i}
  \DisplayProof
\end{equation}
Since we have derived a chain of implications $x_1$, $x_1\limpl x_2$, \ldots, $x_{m\ell-1}\limpl x_{m\ell-1}$, $\olnot{x_{m\ell-1}}$ we can complete the refutation in $\bigoh{m\ell}$ more steps.
Let us record our discussion.

\begin{lemma}
  \label{lem:ub}
  $F_{\ell,m,n}$ has a resolution refutation of length $\bigoh{\ell\cdot(m+n)}$.
\end{lemma}

Before we prove the lower bound let us discuss informally what are the natural
ways to refute this formula in RML, so that we understand which
 behaviours we need to rule out.

If we try to reproduce the previous resolution refutation,
since we cannot reuse the clauses representing $w_{j,1} = w_{j,n}$
because they are not merges, we have to rederive them each time we need
them, which means that it takes $\bigoh{mn\ell}$ steps to derive
the chain of implications $x_1, \ldots,
x_i \limpl x_{i+1}, \ldots, \olnot{x_{n\ell-1}}$.
We call this approach \newrefutation{}.
This refutation has merges (over $w_{\hat\imath,1}$, $x_{i-1}$, and
$x_i$) when we produce
$w_{\hat\imath,1}^b \lor \olnot{x_{i-1}} \lor x_i$, and (over
$x_{i-1}$ and $x_i$) when we produce $\olnot{x_{i-1}} \lor x_i$, but
since we never reuse these clauses the refutation is in fact
tree-like.

An alternative approach, which we call \newrefutation{}, is to start
working with the $\XX$ axioms instead. In this proof we clump together all of
the repeated constraints of the form $w_{j,1} \neq w_{j,n}$ for every
$j\in[\ell]$, and then resolve them out in one go. In other words, we
first derive the sequence of constraints
\begin{align}
  D_i = \biggl(\bigvee_{\hat\imath\in[\min(i,\ell)]} w_{\hat\imath,1}\neq
  w_{\hat\imath,n}\biggr) \lor x_i && \text{for $i\in[m\ell]$} \eqcomma
\end{align}
  where $D_i$ can be obtained from $D_{i-1}$ and the pair of  $\XX$ axioms $A_{i,b}$, then resolve
away the inequalities from $D_{m\ell} = \bigvee_{j\in[\ell]}w_{j,1} \neq w_{j,n}$ using the $\WW$ axioms. However, representing any of the constraints $D_i$ for $i\geq\ell$ requires $2^{\ell}$ clauses, which is significantly larger than $mn\ell$ and even superpolynomial
for large enough $\ell$, so this refutation is not efficient either.
Note that this refutation has merges (over $W$ variables) each time that we
derive $D_i$ with $i \geq \ell$.

A third and somewhat contrived way to build a refutation is to derive
the pair of clauses representing $w_{j,1} = w_{j,n}$ using a
derivation whose last step is a merge, so that they can be
reused. Each of these clauses can be derived individually in $\bigoh{mn\ell}$ steps, for a total of $\bigoh{mn\ell^2}$ steps,
by slightly adapting refutation~\ref{ref:1}, substituting each
derivation of $x_i \limpl x_{i+1}$ by a derivation of
$w_{j,1} \lor \olnot{w_{j,n}} \lor \olnot{x_i} \lor x_{i+1}$ whenever
$i \equiv j \pmod{\ell}$ so that at the end we obtain
$w_{j,1} \lor \olnot{w_{j,n}}$ instead of the empty clause.
Such a
substitution clause can be obtained, e.g., by resolving
$w_{j,1} \lor w_{j,2} \lor \olnot{x_i} \lor x_{i+1}$ with
$\olnot{w_{j,2}} \lor \olnot{w_{j,n}} \lor \olnot{x_i} \lor
x_{i+1}$ as follows
\begin{equation}
  \label{eq:ref-xw}
    \def\ScoreOverhang{2pt}
    \def\defaultHypSeparation{~}
    \small
  \Axiom{w_{j,2} \lor \olnot{w_{j,3}}}
  \Axiom{w_{j,3} \lor \olnot{w_{j,4}}}
  \BinaryInf{w_{j,2} \lor \olnot{w_{j,4}}}
  \UnaryInf{\vdots}
  \UnaryInf{w_{j,2} \lor \olnot{w_{j,n-1}}}
  \Axiom{w_{j,n-1} \lor \olnot{w_{j,n}}}
  \BinaryInf{w_{j,2} \lor \olnot{w_{j,n}}}
  \Axiom{w_{\hat\imath,1} \lor w_{\hat\imath,n} \lor \olnot{x_{i-1}} \lor x_i}
  \BinaryInf{w_{\hat\imath,1} \lor w_{\hat\imath,2} \lor \olnot{x_{i-1}} \lor x_i}
  \Axiom{w_{\hat\imath,1} \lor \olnot{w_{\hat\imath,2}}}
  \Axiom{\olnot{w_{\hat\imath,1}} \lor \olnot{w_{\hat\imath,n}} \lor \olnot{x_{i-1}} \lor x_i}
  \BinaryInf{\olnot{w_{\hat\imath,2}} \lor \olnot{w_{\hat\imath,n}} \lor \olnot{x_{i-1}} \lor x_i}
  \BinaryInf{w_{\hat\imath,1} \lor \olnot{w_{\hat\imath,n}} \lor \olnot{x_{i-1}} \lor x_i}
  \DisplayProof
\end{equation}
After deriving $w_{j,1} = w_{j,n}$ as merges we follow the
next steps of refutation~\ref{ref:1} and complete the refutation in
$\bigoh{m\ell}$ steps. We call this \newrefutation{}.

Observe that the minimum length of deriving the clauses representing
$w_{j,1} = w_{j,n}$ is only $\bigoh{n}$, even in RML, so if we only
used the information that refutation~\ref{ref:3} contains these clauses we would
only be able to bound its length by
$\bigomega{\ell\cdot(m+n)}$. Therefore when we compute the hardness of
deriving a clause we need to take into account not only its semantics
but how it was obtained syntactically.

\subsection{Lower Bound}

Before we begin proving our lower bound in earnest we make two useful observations.

\begin{lemma}
  \label{lem:w-no-merge}
  Let $\eta$ be a resolution derivation that only depends on the $\WW$ axioms. Then $\eta$
  does not contain any merges, and all clauses are supported on $W$.
\end{lemma}

\begin{proof}
  We prove by induction that every clause in $\eta$ is of the form $w_{j,k} \lor \olnot{w_{j,k'}}$ with $k \neq k'$. This is true for the axioms. By induction hypothesis, a generic resolution step over $w_{j,k}$ is of the form
  \begin{equation}
    \Axiom{w_{j,k} \lor \olnot{w_{j,k'}}}
    \Axiom{\olnot{w_{j,k}} \lor w_{j,k''}}
    \BinaryInf{w_{j,k''} \lor \olnot{w_{j,k'}}}
    \DisplayProof
  \end{equation}
and in particular is not a merge.
\end{proof}

\begin{lemma}
   \label{lem:need-all-x}
  Let $\eta$ be a resolution derivation of a clause $C$ supported on
  $W$ variables that uses an $\XX$ axiom. Then $\eta$ uses at least one
  $A_{i,b}$ axiom for each $i\in[m\ell]$.
\end{lemma}

\begin{proof}
  We prove the contrapositive and assume that there is an axiom
  $A_{i,b}$ that is used, and either both $A_{i+1,0}$ and $A_{i+1,1}$ are
  not used, or both $A_{i-1,0}$ and $A_{i-1,1}$ are not. In the first case
  the literal $x_i$ appears in every clause in the path from $A_{i,b}$
  to $C$, contradicting that $C$ is supported on $W$
  variables. Analogously with literal $\olnot{x_{i-1}}$ in the second
  case.
\end{proof}

Our first step towards proving the lower bound is to rule out that refutations like refutation~\ref{ref:2} can be small, and to do
that we show that wide clauses allow for very little progress. This is
a common theme in proof complexity, and the standard tool is to apply
a random restriction to a short refutation in order to obtain a narrow
refutation. However, RMA is not closed under restrictions, as we prove later in Corollary~\ref{cor:restrictions}, and
because of this we need to argue separately about which merges
are preserved.

Let us define the class of restrictions that we use and which need to respect the structure of the formula.
A restriction is an autarky~\cite{MS85Solving} with respect to a set
of clauses $D$ if it satisfies every clause that it touches; in other
words for every clause $C\in D$ either $\restrict{C}{\rho}=1$ or
$\restrict{C}{\rho}=C$.
A restriction is $k$-respecting if it is an autarky with respect to
$\WW$ axioms, we have $\restrict{F_{\ell,m,n}}{\rho} \cong F_{k,m,n}$ up to
variable renaming, and every $X$ variable is mapped to an $X$
variable.
Our definition of a narrow clause is also tailored to the formula at
hand, and counts the number of different $W$-blocks that a clause $C$
mentions. Formally
$\mu(C) = \setsize{\setdescr{j\in[\ell]}{\exists x_{j,k} \in
    \vars{C}}}$.

\begin{lemma}
  \label{lem:random-restriction}
  Let $\pi$ be a resolution refutation of $F_{\ell,m,n}$ of length
  $L = \littleoh{(4/3)^{\ell/8}}$. There exists an $\ell/4$-respecting
  restriction $\rho$ such that every clause in $\restrict{\pi}{\rho}$
  has $\mu(C) \leq \ell/8$.
\end{lemma}

\begin{proof}
We use the probabilistic method.
Consider the following distribution $\mathcal J$ over
$\set{0,1,*}^\ell$: each coordinate is chosen independently with
$\Pr[J_i=0]=\Pr[J_i=1]=1/4$, $\Pr[J_i=*]=1/2$. Given a random variable $J\sim\mathcal J$ sampled according to this distribution, we derive a random restriction $\rho$ as follows:
$\rho(w_{j,i})=J_j$, $\rho(x_i)=*$ if $J_{\hat\imath}=*$, and
$\rho(x_i)=\rho(x_{i-1})$ otherwise (where $\rho(x_0)=1$).

Observe that
$\restrict{F_{\ell,m,n}}{\rho} \cong F_{\setsize{J^{-1}(*)},m,n}$ up to
variable renaming, and by a Chernoff bound we have
$\Pr[\setsize{J^{-1}(*)}<\ell/4]\leq e^{-\ell/16}$.

We also have, for every clause $C\in\pi$ with $\mu(C)>\ell/8$, that
\begin{equation}
  \Pr[\restrict{C}{\rho} \neq 1] \leq (3/4)^{\mu(C)} \leq (3/4)^{\ell/8} \eqperiod
\end{equation}
Therefore by a union bound the probability that
$\setsize{J^{-1}(*)}<\ell/4$ or that any clause has
$\mu(\restrict{C}{\rho})>\ell/8$ is bounded away from $1$ and we conclude
that there exists a restriction $\rho$ that satisfies the conclusion
of the lemma.
\end{proof}

Note that $s(\restrict{\pi}{\rho})$ is a resolution refutation of
$\restrict{F_{n,\ell}}{\rho}$, but not necessarily a RMA
refutation, therefore we lose control over which clauses may be
reused\footnote{Recall that $s(\pi)$ is the syntactic equivalent of $\pi$.}. Nevertheless, we can identify a fragment of
$s(\restrict{\pi}{\rho})$ where we still have enough information.

\begin{lemma}
  \label{lem:psi-no-reuse}
  There exists an integer $t$ such that $\psi=s(\restrict{\pi[1,t]}{\rho})$ is a
  resolution derivation of a clause supported on $W$ variables that
  depends on an $\XX$ axiom and where no clause supported on $W$
  variables is reused.
\end{lemma}

\begin{proof}
  Let $C_t\in\pi$ be the
  first clause that depends on an $\XX$ axiom and such that
  $D_t=s(\restrict{C_t}{\rho})$ is supported on $W$, which exists
  because $\bot$ is one such clause.

  By definition of $t$, we have that every ancestor $D_k\in\psi$ of
  $D_t$ that is supported on $W$
  variables corresponds to a clause $C_k$ in $\pi$ that only
  depends on $\WW$ axioms, hence by Lemma~\ref{lem:w-no-merge} $C_k$ is
  not a merge. By definition of RMA $C_k$ is not reused, and by
  construction of $s(\cdot)$ neither is $D_k$.

  It remains to prove that $D_t$ depends on an $\XX$ axiom. Since $C_t$
  depends on an $\XX$ axiom, at least one of its predecessors $C_p$ and
  $C_q$ also does, say $C_p$. By definition of $t$,
  $D_p=s(\restrict{C_p}{\rho})$ is not supported on $W$, and hence by
  Lemma~\ref{lem:w-no-merge} either $D_p$ depends on an $\XX$ axiom or
  $D_p=1$. Analogously, if $C_q$ also depends on an $\XX$ axiom then so
  does $D_q=s(\restrict{C_j}{\rho})$ (or it is $1$) and we are
  done.
  Otherwise $C_q$ is of the form $w_{j,k} \lor \olnot{w_{j,k'}}$
  and is either satisfied by $\rho$ or left untouched. In both cases
  we have that $D_q \not\vDash \restrict{C_t}{\rho}$ (trivially in the
  first case and because $D_q$ contains the pivot while $C_t$ does not
  in the second), hence $D_t$ depends on $D_p$.
\end{proof}

Note that $C_t$ may be semantically implied by the $\WW$ axioms, and
have a short derivation as in refutation~\ref{ref:3}, therefore we are forced to use syntactic arguments to argue that deriving $C_t$
\emph{using an $\XX$ axiom} takes many resolution steps.

The next step is to break $\psi$ into $m$ (possibly intersecting)
parts, each corresponding roughly to the part of $\psi$ that uses $\XX$
axioms with variables in an interval of length $\ell$ (by Lemma~\ref{lem:need-all-x} we can assume that $\psi$ contains axioms from every interval).
To do this we use the following family of restrictions defined for
$i\in[n]$:
\begin{align}
  \sigma_i(x_{i'}) &=
  \begin{cases*}
    1 & if $i' \leq i\ell$\\
    * & if $i\ell < i' \leq (i+1)\ell$\\
    0 & if $(i+1)\ell < i'$
  \end{cases*} &
  \sigma_i(w_{i',j}) &= *
\end{align}
Let $X_i = X \intersection \sigma_i^{-1}(*)$ and note that
$\restrict{F_{\ell,m,n}}{\sigma_i} \cong F_{\ell,1,n}$.

Clauses in $\psi$ with many $X$ variables could be tricky to classify, but
intuitively it should be enough to look at the smallest positive
literal and the largest negative literal, since these are the hardest
to eliminate.
Therefore we define $r(C)$ to be the following operation on a clause:
literals over $W$ variables are left untouched, all positive $X$
literals but the smallest are removed, and all negative $X$ literals
but the largest are removed. Formally,
\begin{equation}
 r\biggl(\Lor_{i\in A} x_i \lor \Lor_{i\in B} \olnot{x_i} \lor \Lor_{(i,j)\in C} w_{i,j}^{b_{i,j}}\biggr) = x_{\min A} \lor \olnot{x_{\max B}} \lor \Lor_{(i,j)\in C} w_{i,j}^{b_{i,j}}
\end{equation}
where $x_{\min A}$ (resp. $\olnot{x_{\max B}}$) is omitted if $A$ (resp. $B$) is empty.

We need the following property of $r(C)$.
\begin{lemma}
  \label{lem:r}
  If $\restrict{C}{\sigma_i}\neq 1$ and $\vars{r(C)} \intersection X_i = \emptyset$ then $\restrict{C}{\sigma_i}$ is supported over $W$ variables.
\end{lemma}

\begin{proof}
  The hypothesis that $\vars{r(C)} \intersection X_i = \emptyset$
  implies that the smallest positive $X$ literal in $C$ is either not
  larger than $i\ell$ or larger than $(i+1)\ell$, but the hypothesis
  that $\restrict{C}{\sigma_i} \neq 1$ rules out the first
  case. Therefore all positive $X$ literals are falsified by
  $\sigma_i$.
  Analogously the largest negative $X$ literal is not larger than
  $i\ell$ and all negative $X$ literals are also falsified.
\end{proof}

We define each part $\psi_i$ to consist of all clauses $C\in\sigma$ such that $C$ is
\begin{enumerate}
\item an $\XX$ axiom not satisfied by $\sigma_i$; or
\item the conclusion of an inference with pivot in $X_i$; or
\item the conclusion of an inference with pivot in $W$ that depends on an $\XX$ axiom if $r(C)$ contains a variable in $X_i$; or
\item the conclusion of an inference with pivot in $W$ that does not depend on $\XX$ axioms if the \emph{only} immediate successor of $C$ is in $\psi_i$.
\end{enumerate}
This is the point in the proof where we use crucially that the
original derivation is in RMA form: because clauses that do not depend
on $\XX$ axioms are not merges, they have only one successor and the
definition is well-formed.

Ideally we would like to argue that parts $\psi_i$ are pairwise
disjoint. This is not quite true, but nevertheless they do not overlap
too much.

\begin{lemma}
  \label{lem:overcounting}
  Let $\psi$ and $\setdescr{\psi_i}{i\in[\ell]}$ be as discussed above. Then
  $2\setsize{\psi}\geq \sum_i\setsize{\psi_i}$.
\end{lemma}

\begin{proof}
  Axioms may appear in at most two different $\psi_i$,
  and clauses
  obtained after resolving with an $X$ pivot in only one.  The only
  other clauses that depend on an $\XX$ axiom and may appear in
  different $\psi_i$ are obtained after resolving with a $W$ pivot,
  but since $r(C)$ only contains two $X$ variables, such clause only
  may appear in two different $\psi_i$. Finally, clauses that do not
  depend on an $\XX$ axiom appear in the same $\psi_i$ as one clause
  of the previous types, and therefore at most two different parts.
\end{proof}

To conclude the proof we need to argue that each $\psi_i$ is
large. The intuitive reason is that $\psi_i$ must use one $\XX$ axiom
for each $j\in [(i\ell,(i+1)\ell]$, which introduces a pair of $W$
variables from each $W_j$ block, but since no clause contains more than
$\ell/8$ such variables, we need to use enough $\WW$ axioms to remove
the aforementioned $W$ variables. Formally the claim follows from
these two lemmas.

\begin{lemma}
  \label{lem:psi-valid}
  For each $i\in[\ell]$ there exists an integer $t_i$ such that
  $s(\restrict{\psi_i[1,t_i]}{\sigma_i})$ is a resolution derivation
  of a clause supported on $W$ variables that depends on an $\XX$
  axiom.
\end{lemma}

\begin{proof}
  Let $C_{t_i}$ be the first clause in $\psi_i$ that depends on an $\XX$
  axiom and such that  $\restrict{C_{t_i}}{\sigma_i}$ is supported on $W$ variables. We
  prove that $t_i$ is well-defined, that
  $\restrict{\psi_i[1,t_i]}{\sigma_i}$ is a valid semantic resolution
  derivation, and that $D_{t_i} = s(\restrict{\psi_i}{\sigma_i})$
  depends on an $\XX$ axiom.

  Our induction hypothesis is that for $k\leq t_i$ (or any $k$ if
  $t_i$ does not exist), if the clause $C_k \in \psi$ depends on an
  $\XX$ axiom and is not satisfied by $\sigma_i$, then there exists a
  clause $C_{k'} \in \psi_i$ with $k'\leq k$ that implies $C_k$ modulo
  $\sigma_i$, that is
  $\restrict{C_{k'}}{\sigma_i} \vDash \restrict{C_k}{\sigma_i}$, and
  depends on an $\XX$ axiom (over $\psi$).

  If the induction hypothesis holds then $t_i$ is well-defined: since
  $C_t$ is not satisfied by $\sigma_i$ and depends on an $\XX$ axiom
  there exists a clause $C_{t'}\in \psi_i$ that depends on an $\XX$
  axiom and such that
  $\restrict{C_{t'}}{\sigma_i} \vDash \restrict{C_t}{\sigma_i} = C_t$,
  which is supported on $W$ variables.

  The base case is when $C_k$ is a non-satisfied $\XX$ axiom, where we can take
  $C_{k'}=C_k$.
  For the inductive case let $C_p$ and $C_q$ be the premises of $C_k$
  in $\psi$. If exactly one of the premises, say $C_p$, is
  non-satisfied and, furthermore, depends on an $\XX$ axiom, then by
  the induction hypothesis we can take $C_{k'}=C_{p'}$.  Otherwise we
  need to consider a few subcases.
  If the pivot is an $X$ variable then both premises depend on an
  $\XX$ axiom (by Lemma~\ref{lem:w-no-merge}), hence neither premise
  is satisfied. It follows that the pivot is unassigned by $\sigma_i$,
  and therefore we can take $C_{k'}=C_k$.

  If the pivot is a $W$ variable then, because $\sigma_i$ only assigns
  $X$ variables, neither premise is satisfied. We have two subcases:
  if exactly one premise depends on an
  $\XX$ axiom, say $C_p$, then $C_{p'}$ is present in $\psi_i$, and by
  construction of $\psi_i$ the other premise $C_q$ is present in $\psi_i$
  if and only if the conclusion $C_k$ is. If both premises depend on an $\XX$ axiom
  then both $C_{p'}$ and $C_{q'}$ are present in $\psi_i$.

  Therefore in the two latter subcases it is enough to prove that
  $C_k \in \psi_i$, since then we can take $C_{k'}=C_{k}$ and we have
  that $\restrict{C_k}{\sigma_i}$ follows from a valid semantic
  resolution step.
  Indeed by Lemma~\ref{lem:r}
  $\restrict{C_k}{\sigma_i}$ is a clause supported on $W$ variables,
  which by definition of $C_{t_i}$ implies that $k=t_i$.
  However, since the pivot is a $W$ variable,
  $\restrict{C_{p'}}{\sigma_i}$ is also supported on $W$
  variables and, together with the fact that $C_{p'}$ depends on an
  $\XX$ axiom, this contradicts that $C_{t_i}$ is the first such clause.

  This finishes the first induction argument and proves that
  $\restrict{\psi[1,t_i]}{\sigma_i}$ is a valid semantic derivation;
  it remains to prove that $D_{t_i}$ depends on an $X$ axiom over
  $s(\restrict{\psi_i}{\sigma_i})$.
  We prove by a second induction argument that for every clause $D_k \in s(\restrict{\psi_i[1,t_i]}{\sigma_i})$, if $C_k$ depends
  on an $\XX$ axiom then so does $D_k$.
  The base case, when $D_k$ is an axiom, holds.

  For the inductive case
  fix $C_k$, $E_k=\restrict{C_k}{\sigma_i}$, and $D_k=s(E_k)$, and let
  $E_p=\restrict{C_p}{\sigma_i}$ and
  $E_q=\restrict{C_q}{\sigma_i}$ be the premises of $E_k$ in
  $\restrict{\psi_i}{\sigma}$.
  When both $C_p$ and $C_q$ depend on an $X$ axiom, then by
  hypothesis so do $D_p$ and $D_q$ and we are done. We only need to
  argue the case when one premise $C_p$ depends on an $X$ axiom and
  the other premise $C_q$ does not. In that case, because
  $\sigma_i$ only affects $X$ variables, all the axioms used in the
  derivation of $C_q$ are left untouched by $\sigma_i$, therefore
  we have that $s(\sigma_i(C_q))=C_q$, which contains the pivot used
  to derive $C_k$ and therefore does not imply $s(\sigma_i(C_k))$. By
  construction of $s(\cdot)$, $s(\sigma_i(C_k))$ depends on
  $s(\sigma_i(C_p))$.
\end{proof}

\begin{lemma}
  \label{lem:psi-large}
  Let $\eta$ be a resolution derivation from $F_{\ell,1,n}$ of a
  clause $C$ supported on $W$ variables that depends on an $\XX$
  axiom. Then $\setsize{\eta} \geq (n-2)(\ell-\mu(C))/2$.
\end{lemma}

\begin{proof}
  By Lemma~\ref{lem:need-all-x} we can assume that $\eta$ uses at
  least one $A_{j,b}$ axiom for each $j\in[\ell]$.

  Let $J=\setdescr{j\in[\ell]}{\exists w_{j,k}\in\vars{C}}$ be the set of $W$
  blocks mentioned by $C$. We show that for each $j\in\olnot{J}=[\ell]\setminus J$ at
  least $(n-2)/2$ axioms over variables in $W_j$ appear in $\eta$, which
  makes for at least $(n-2)\setsize{\olnot{J}}/2 = (n-2)(\ell-\mu(C))/2$ axioms.

  Fix $j\in\olnot{J}$ and assume for the sake of contradiction that
  less than $(n-2)/2$ axioms over variables in $W_j$ appear in
  $\eta$. Then there exists $k\in[2,n-1]$ such that variable $w_{j,k}$
  does not appear in $\eta$. Rename variables as follows:
  $w_{j,k'} \mapsto y_{k'}$ for $k'<k$, and
  $w_{j,k'} \mapsto \olnot{y_{k'-n}}$ for $k'>k$. Then we can prove by
  induction, analogously to the proof of Lemma~\ref{lem:w-no-merge},
  that every clause derived from axiom $A_{j,b}$ is of the form
  $y_{k'} \lor \olnot{y_{k''}} \lor D$ where $D$ are literals
  supported outside $W_j$. Since that includes $C$, it contradicts our
  assumption that $j\notin J$.
\end{proof}

To conclude the proof of Theorem~\ref{th:separation} we simply need to
put the pieces together.

\begin{proof}[Proof of Theorem~\ref{th:separation}]
  We take as the formula family $F_{\ell=48\log n,n,n}$, for which a
  resolution refutation of length $\bigoh{n\log n}$ exists by
  Lemma~\ref{lem:ub}.

  To prove a lower bound we and assume that a RMA refutation $\pi$ of
  length $L \leq n^3 = 2^{16\ell} = \littleoh{(4/3)^{8\ell}}$ exists;
  otherwise the lower bound trivially holds. We apply the
  restriction given by Lemma~\ref{lem:random-restriction} to $\pi$ and
  we use Lemma~\ref{lem:psi-no-reuse} to obtain a resolution
  derivation $\psi$ of a clause supported on $W$ variables that uses
  an $\XX$ axiom. We then break $\psi$ into $m$ parts $\psi_i$, each
  of size at least $n\ell/16$ as follows from
  Lemmas~\ref{lem:psi-valid} and~\ref{lem:psi-large}. Finally by
  Lemma~\ref{lem:overcounting} we have
  $\setsize{\pi} \geq \setsize{\psi} \geq mn\ell/32 =
  \bigomega{n^2\log n}$.
\end{proof}

\subsection{Structural Consequences}

Theorem~\ref{th:separation} immediately gives us two
structural properties of RML and RMA. One is that proof length may
decrease when introducing a weakening rule.
\begin{corollary}
  \label{cor:weakening}
  There exists a family of formulas over $\bigoh{n\log n}$
  variables and $\bigoh{n\log n}$ clauses that have RML with
  weakening refutations of length $\bigoh{n\log n}$ but every RMA refutation requires
  length $\bigomega{n^2\log n}$.
\end{corollary}

\begin{proof}
  Consider the formula $F_n \land \olnot{z}$, where $F_n$ is the
  formula given by Theorem~\ref{th:separation} and $z$ is a new
  variable. If we weaken every clause $C\in F_n$ to $C \lor z$ then we
  can derive $F \lor z \vdash z$ in $\bigoh{n\log n}$ RML steps
  because each inference is a merge. However, if we cannot do
  weakening, then $\olnot{z}$ cannot be resolved with any clause in
  $F_n$ and the lower bound of Theorem~\ref{th:separation} applies.
\end{proof}

The second property is that RML and RMA are not \emph{natural} proof
systems in the sense of~\cite{BKS04TowardsUnderstanding} because proof
length may increase after a restriction.
\begin{corollary}
  \label{cor:restrictions}
  There exists a restriction $\rho$ and a family of formulas
  over $\bigoh{n\log n}$ variables and $\bigoh{n\log n}$ clauses that
  have RML refutations of length $\bigoh{n\log n}$ but every RMA refutation of
  $\restrict{F_n}{\rho}$ requires length $\bigomega{n^2\log n}$.
\end{corollary}

\begin{proof}
  Consider the formula $G_n = (F_n \lor z) \land \olnot{z}$, where
  $F_n$ is the formula given by Theorem~\ref{th:separation},
  $F \lor z = \setdescr{C \lor z}{C\in F}$, and $z$ is a new
  variable. As in the proof of Corollary~\ref{cor:weakening} there is
  a RML derivation of $z$ of length $\bigoh{n\log n}$ steps, while
  $\restrict{G_n}{\rho}=F_n$.
\end{proof}

\section{Further Separations}
\label{sec:separationgeneral-x}

We can separate the different flavours of merge resolution that we introduced using a few variations of
$F_{\ell,m,n}$ where we add a constant number of redundant clauses for
each $i\in[\ell]$. We consider these different clauses part of $\WW$.

Upper bounds all follow the same pattern. We first show on a case-by-case basis how to obtain $w_1\olnot{w_n}$ and $\olnot{w_1}w_n$ as lemmas, and then proceed as in Section~\ref{sec:ub}.

Towards proving lower bounds we are going to generalize the lower
bound part of the proof of Theorem~\ref{th:separation} to apply to
these variations as well. Fortunately we only require a few local modifications.

First, we need to
prove an equivalent of Lemma~\ref{lem:w-no-merge}, which we do on a
case-by-case basis.

Second, we need to show that $k$-respecting restrictions can be
extended to the new variables. For each block $J_i$, since the new
clauses are semantically subsumed by $w_{i,1}=w_{i,2}$, there exists
a way to map the new variables into $w_{i,1}$ and $w_{i,2}$ so that the result
of the restriction is the same as if we had started with clauses
$\olnot{w_{i,1}} \lor w_{i,2}$ and $w_{i,1} \lor \olnot{w_{i,2}}$, which are already part of $\WW_i$. That is, the formula that we work with after Lemma~\ref{lem:psi-no-reuse} is a copy of an unaltered $F_{\ell',m',n'}$ formula.

The only part of the lower bound that depends on the specific
subsystem of Resolution is Lemma~\ref{lem:psi-no-reuse}; afterwards
all the information we use is that no clause supported on $W$
variables is reused. Furthermore, the only property of the subsystem
that we use in the proof of Lemma~\ref{lem:psi-no-reuse} is that
Lemma~\ref{lem:w-no-merge} applies. Therefore, the modifications we
just outlined are sufficient for the lower bound to go through.

\subsection{Separation between RMA and LRMA}

\begin{proposition}
  \label{prop:rma-srrma-x}
  There exists a family of formulas over $\bigoh{n\log n}$ variables and
  $\bigoh{n\log n}$ clauses that have RMA refutations of length
  $\bigoh{n\log n}$ but every LRMA refutation requires length
  $\bigomega{n^2 \log n}$.
\end{proposition}

The separating formula is $F_{m,n,\ell}^{(1)}$, where we add to $F_{m,n,\ell}$ clauses
  \begin{align}
  &\olnot{w_{i,1}} \lor w_{i,2} \lor \olnot {z_i}, \tag{C1}\\
  &\olnot{w_{i,2}} \lor z_i, \tag{C2}\\
  &w_{i,1} \lor \olnot{w_{i,2}} \lor \olnot {y_i}, \tag{C3}\\
  &w_{i,2} \lor y_i, \tag{C4}
\end{align}
  for each $i\in[\ell]$.
The new variables can be assigned as $z_{i} = w_{i,1}$ and $y_{i} = \olnot{w_{i,1}}$ to obtain the original formula back.

The upper bound follows from the following lemma.

\begin{lemma}
  Clauses $w_{i,1} \lor \olnot{w_{i,n}}$ and $\olnot{w_{i,1}} \lor w_{i,n}$ can be derived as lemmas from $F_{m,n,\ell}^{(1)}$ in length $\bigoh{n}$ in RMA.
\end{lemma}

\begin{proof}
  We resolve clause $\olnot{w_{i,1}} \lor w_{i,2}$ first with (C2) and then (C1) in order to obtain $\olnot{w_{i,1}} \lor w_{i,2}$ as a merge, then derive
  $\olnot{w_{i,1}} \lor w_{i,n}$, having a merge as its ancestor, so it can be
  remembered. Analogously starting from $w_{i,1} \lor \olnot{w_{i,2}}$, (C3), and (C4) we can obtain $w_{i,1} \lor \olnot{w_{i,n}}$ as a lemma.
\end{proof}

The following observation is useful for the lower bound.

\begin{lemma}
  \label{lem:sr-no-clashing}
  Let $C$ and $D$ be clauses with two pairs of opposite literals. Then $C$ and $D$ cannot appear in the same locally regular input derivation.
\end{lemma}

\begin{proof}
  Let $C=x\lor y \lor C'$ and $D=\olnot{x} \lor \olnot{y} \lor D'$. Assume wlog that $C$ is the first clause out of $C$ and $D$ to appear in the derivation. If $x$ or $y$ are used as pivots before $D$, then the locally regular condition prevents using $D$ as an axiom. Otherwise $x \lor y$ appears in the derivation since the time $C$ is used, which also prevents using $D$.
\end{proof}

The equivalent of Lemma~\ref{lem:w-no-merge} is the following.

\begin{lemma}
  \label{lem:srrma-no-reuse}
  Let $\eta$ be a LRMA derivation that only depends on $\WW$ axioms. Then no clause in $\eta$ can be reused.
\end{lemma}

\begin{proof}
  We can only obtain a merge using one of (C1) or (C3), assume wlog
  (C1) is the first of these to be used in the derivation. By
  Lemma~\ref{lem:sr-no-clashing} neither (C2) nor (C3) appear in the
  derivation. We can show by induction that we can only obtain clauses
  of the form $\olnot{w_{i,j}} \lor \olnot{w_{i,j'}} \lor \olnot{z_i}$ or
  $y_i \lor \olnot{w_{i,j}} \lor \olnot{z_i}$, never as a merge.
\end{proof}

\subsection{Separation between RML/LRMA and LRML}

\begin{proposition}
  \label{prop:rml-srrma-srrml-x}
  There exists a family of formulas over $\bigoh{n\log n}$ variables and
  $\bigoh{n\log n}$ clauses that have RML and LRMA and refutations of length
  $\bigoh{n\log n}$ but every LRML refutation requires length
  $\bigomega{n^2 \log n}$.
\end{proposition}

The separating formula is $F_{m,n,\ell}^{(2)}$, where we add to $F_{m,n,\ell}$ clauses
\begin{align}
  &z_i \lor \olnot{w_{i,1}} \lor w_{i,2}, \tag{C1}\\
  &\olnot{z_i} \lor \olnot{w_{i,1}} \lor w_{i,2}, \tag{C2}\\
  &y_i \lor w_{i,1} \lor \olnot{w_{i,2}}, \tag{C3}\\
  &\olnot{y_i} \lor w_{i,1} \lor \olnot{w_{i,2}}, \tag{C4}
\end{align}
for each $i\in[\ell]$.
The new variables can be assigned as $z_{i} = 1$ and $y_{i} = 1$ to obtain the original formula back.

The upper bounds follow respectively from the following lemmas.

\begin{lemma}
  Clauses $w_{i,1} \lor \olnot{w_{i,n}}$ and $\olnot{w_{i,1}} \lor w_{i,n}$ can be derived as lemmas from $F_{m,n,\ell}^{(2)}$ in length $\bigoh{n}$ in RML.
\end{lemma}

\begin{proof}
  We first resolve clauses $\olnot{w_{n-1}} \lor w_{i,n}$, $\olnot{w_{i,n-2}} \lor w_{i,n-1}$, \dots, $\olnot{w_{i,2}} \lor w_{i,3}$, (C1) to obtain $z_i \lor \olnot{w_{i,1}} \lor w_{i,n}$. We continue the input derivation resolving with (C2) to obtain $\olnot{w_{i,1}} \lor w_{i,2} \lor w_{i,n}$. We then resolve with $\olnot{w_{i,2}} \lor w_{i,3}$, $\olnot{w_{i,3}} \lor w_{i,4}$, \dots, $\olnot{w_{i,n-1}} \lor w_{i,n}$ to obtain $\olnot{w_{i,1}} \lor w_{i,n}$ as a merge over $w_{i,n}$. Analogously we can obtain $w_{i,1} \lor \olnot{w_{i,n}}$.
\end{proof}

\begin{lemma}
  Clauses $w_{i,1} \lor \olnot{w_{i,n}}$ and $\olnot{w_{i,1}} \lor w_{i,n}$ can be derived as lemmas from $F_{m,n,\ell}^{(2)}$ in length $\bigoh{n}$ in LRMA.
\end{lemma}

\begin{proof}
  We resolve clauses (C1) and (C2) to obtain
  $\olnot{w_{i,1}} \lor w_{i,2}$, which is a merge, then derive
  $\olnot{w_{i,1}} \lor w_{i,n}$, having a merge as its ancestor, so it can be
  used as a lemma. Analogously starting from (C3) and (C4) we can obtain $w_{i,1} \lor \olnot{w_{i,n}}$ as a lemma.
\end{proof}

The equivalent of Lemma~\ref{lem:w-no-merge} is the following.

\begin{lemma}
  \label{lem:srrml-no-reuse}
  Let $\eta$ be a LRML derivation that only depends on $\WW$ axioms. Then no clause in $\eta$ can be reused.
\end{lemma}

The proof idea is that the only merge we can obtain
involves the $z_i$ or the $y_i$ variable. If we just resolve the two clauses over such a variable we
obtain a clause we already had, so this is useless. Otherwise we are
resolving one of $w_2$ away, which would be
reintroduced at the time of resolving $z_i$ away, and that is not
allowed by the SR condition.

\begin{proof}
  We can only obtain a merge by using one of the new clauses
  (C1)--(C4). If we resolve either pair of clauses over $y_i$ or over
  $z_i$ then we obtain a clause that was already present in the formula,
  and therefore we may preprocess such derivation away.

  Otherwise consider the first step in the derivation where one of the
  new clauses is used as a premise, assume wlog it is (C1). That step is with a clause of the
  form $\olnot{w_{i,2}} \lor w_{i,j}$, and we obtain a clause of the form $z_i \lor \olnot{w_{i,1}} \lor w_{i,j}$,
  which is not a merge. That clause can be possibly resolved over
  $w_{i,j}$ ($j>2$) to obtain other clauses of the same form, neither of
  which is a merge, but it cannot be resolved over $y_i$, $z_i$, or $w_{i,1}$
  because that step would reintroduce variable $w_{i,2}$.
\end{proof}

\subsection{Separation between LRML and REL}

\begin{proposition}
  \label{prop:srrml-rel-x}
  There exists a family of formulas over $\bigoh{n\log n}$ variables and
  $\bigoh{n\log n}$ clauses that have LRML refutations of length
  $\bigoh{n\log n}$ but every REL refutation requires length
  $\bigomega{n^2 \log n}$.
\end{proposition}

The separating formula is $F_{m,n,\ell}^{(3)}$, where we add to $F_{m,n,\ell}$ clauses
\begin{align}
  &\olnot{w_{i,1}} \lor \olnot{w_{i,2}} \lor w_{i,3}, \tag{C1}\\
  &w_{i,1} \lor w_{i,2} \lor \olnot{w_{i,3}} \tag{C2}
\end{align}
for each $i\in[\ell]$.
If we assign $w_{i,2}=w_{i,1}$ we obtain a copy of $F_{m,n-1,\ell}$ which, even if technically it is not the same formula we started with, is enough for our purposes.

The upper bound follows from the following lemma.

\begin{lemma}
  Clauses $w_{i,1} \lor \olnot{w_{i,n}}$ and $\olnot{w_{i,1}} \lor w_{i,n}$ can be derived as lemmas from $F_{m,n,\ell}^{(3)}$ in length $\bigoh{n}$ in LRML.
\end{lemma}

\begin{proof}
  We resolve (C1) with $\olnot{w_{i,3}} \lor w_{i,4}$, \dots, $\olnot{w_{i,n-1}} \lor w_{i,n}$ to obtain $\olnot{w_{i,1}} \lor \olnot{w_{i,2}} \lor w_{i,n}$, then with $\olnot{w_{i,1}} \lor w_{i,2}$ to obtain $\olnot{w_{i,1}} \lor w_{i,n}$ as a merge. Analogously starting from (C2) we can obtain $w_{i,1} \lor \olnot{w_{i,n}}$ as a lemma.
\end{proof}
The equivalent of Lemma~\ref{lem:w-no-merge} is the following.

\begin{lemma}
  \label{lem:greedy-no-reuse}
  Let $\eta$ be a REL derivation that only depends on $\WW$ axioms. Then no clause in $\eta$ can be reused.
\end{lemma}

\begin{proof}
  Observe that every derivable clause has width at least $2$. Let $C$ be any derivable clause and $\ell$ any literal in $C$. We have that $\alpha=\olnot{C\setminus\ell}$ is not empty. However, assigning any variable $w_{i,j}$ immediately propagates all variables, hence $\ell$ is not empowering.
\end{proof}

\section{Concluding Remarks}
\label{sec:remarks}
In this paper, we address the question of the tightness of simulation of resolution proofs by CDCL solvers. Specifically, we show that RMA, among other flavours of DAG-like merge resolution, simulates standard resolution with at most a linear multiplicative overhead. However, contrary to what we see in the tree-like case, this overhead is necessary. While the proof systems we introduce help us explain one source of overhead in the simulation of resolution by CDCL, it is not clear if they capture it exactly. In other words, an interesting future direction would be to explore whether it is possible for CDCL to simulate some flavour of merge resolution with less overhead than what is required to simulate standard resolution.

 \section*{Acknowledgements}
The authors are grateful to Yuval Filmus and a long list of participants in
the program \emph{Satisfiability: Theory, Practice, and Beyond} at the
Simons Institute for the Theory of Computing for numerous discussions.
This work was done in part while the authors were visiting the Simons
Institute for the Theory of Computing.

\appendix
\section{Tree-like Merge Resolution}
\label{sec:extra}

For completeness we informally sketch the proofs of
Lemma~\ref{lem:tree-to-input} and Theorem~\ref{th:tree-to-merge},
which can be found in full detail in~\cite{Andrews68Resolution}.

\begin{lemma}[Lemma~\ref{lem:tree-to-input}, restated]
  \label{app:lem:tree-to-input}
  If there is a tree-like resolution derivation of $C$ of length $L$
  where at most the root is a merge, then there is an input resolution
  derivation of some $C' \subseteq C$ of length at most $L$.
\end{lemma}

\begin{proof}[Proof (sketch)]

  We prove by induction on $\setsize{\eta}$ that for every axiom
  $E\in\eta$ there exists an input derivation of $C'$
that uses a subset of the axioms of $\eta$
  where $E$ is the
  topmost axiom.
  As intermediate objects we allow clauses in this derivation to
  contain opposite literals; these are cleaned up later.

  Let $C=\res(A\lor x, B\lor\olnot{x})$, and let $\eta_1$ and $\eta_2$
  be the derivations used to infer $A\lor x$ and $B\lor x$ respectively. Assume wlog that $E\in\eta_1$.
  Since $\eta_2$ does not contain any merges there exists a unique
  path from $B\lor\olnot{x}$ to an axiom $D\lor \olnot{x} \in \eta_2$, where all clauses contain $\olnot{x}$. Note that other clauses in $\eta_2$ might still contain $x$ or $\olnot{x}$.
  We
  replace $D \lor \olnot{x}$ by $D$ in $\eta_2$ (and consequently remove all the
  occurrences of $\olnot{x}$ in the aforementioned path) and we obtain a valid derivation
  $\eta_3$ of $B$.
  We apply the induction hypothesis to $\eta_1$ and $\eta_3$ to obtain
  two unit derivations $\eta_4$ and $\eta_5$ of $A' \lor x \subseteq A \lor x$
  and $B' \subseteq B$ whose topmost leaves are $E$ and $D$.
  We replace $D$ by $D \lor A'$ in $\eta_5$ and obtain a unit
  derivation $\eta_6$ of $B' \lor A'' \subseteq B' \lor A'$.
  We stitch together $\eta_4$ and $\eta_6$ by observing that
  $\res(A' \lor x, D \lor \olnot{x}) = A' \lor D$, which is the only
  axiom in $\eta_6$ not present in the original axioms, and obtain a
  unit derivation $\eta_7$ of $B\lor A'' = C' \subseteq C$ that only uses original axioms.

  Finally, and outside the inductive argument, we get rid of clauses
  that contain opposite literals by
  replacing any such clause by $1$ to obtain a semantic derivation
  $\eta_8$. Its syntactic counterpart $s(\eta_8)$ satisfies the
  conclusion of the lemma.
\end{proof}

\begin{figure}
  \centering
    \includegraphics{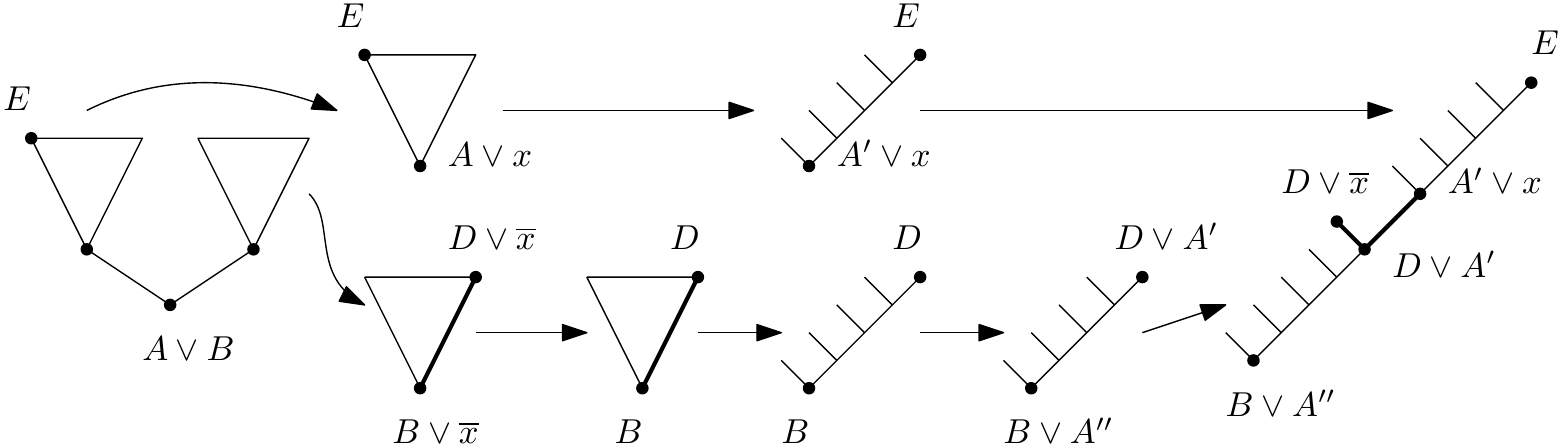}
  \caption{Proof of Lemma~\ref{app:lem:tree-to-input}}
  \label{fig:tree-to-input}
\end{figure}

\begin{theorem}[Theorem~\ref{th:tree-to-merge}, restated]
  \label{app:th:tree-to-merge}
  If there is a tree-like resolution derivation of $C$ of length $L$,
  then there is a merge resolution derivation of some $C' \subseteq C$ of
  length at most $L$.
\end{theorem}

\begin{proof}[Proof (sketch)]
  The proof is by induction on the number of merges. The base case
  when there are no merges follows by
  Lemma~\ref{app:lem:tree-to-input}. Otherwise let $\psi$ be a subtree
  where exactly the root $C$ is a merge. Let $\psi'$ be the input
  resolution derivation of $C'$ given by
  Lemma~\ref{app:lem:tree-to-input}, let $D$ be the last merge in
  $\psi'$, and let $\omega$ and $\omega'$ be the fragments of $\psi'$
  from $D$ to $C'$ and up to $D$ respectively. We replace $\psi$ by
  $\omega$ in $\eta$ to obtain a refutation $\eta'$ that uses $D$ as
  an axiom (note that in replacing $C$ by $C'$ we may have to prune
  away parts of $\eta$). Because $\eta'$ has one less merge we can
  apply the induction hypothesis and obtain a merge resolution
  derivation $\psi''$. Finally we replace the axiom $D$ by the
  derivation $\omega'$.
\end{proof}

\begin{figure}
  \centering
  \iffullversion\newcommand{\theoremone}{theorem1eccc}\else\newcommand{\theoremone}{theorem1}\fi
  \includegraphics{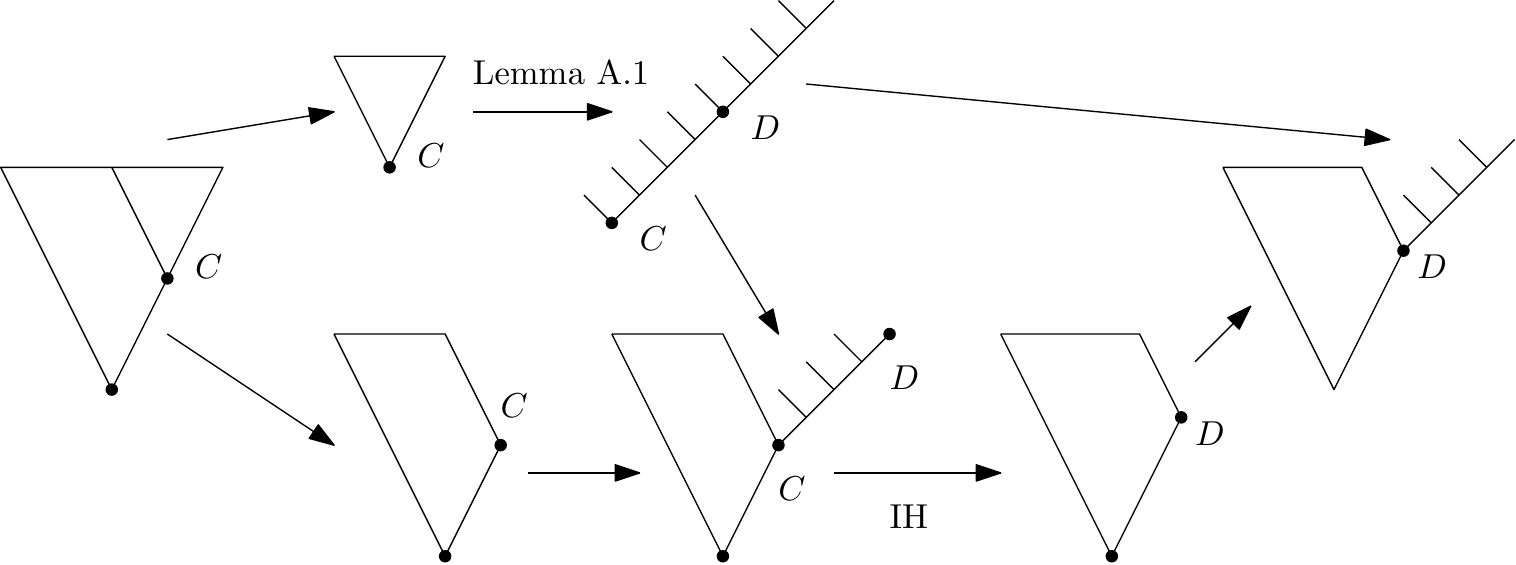}
  \caption{Proof of Theorem~\ref{app:th:tree-to-merge}}
  \label{fig:tree-to-merge}
\end{figure}

\bibliography{refArticles,refLocal}

\newcommand{\etalchar}[1]{$^{#1}$}
\begin{thebibliography}{ZMMM01}

\bibitem[ABH{\etalchar{+}}08]{ABHJS08Conflict}
Gilles Audemard, Lucas Bordeaux, Youssef Hamadi, Sa{\"{\i}}d Jabbour, and
  Lakhdar Sais.
\newblock A generalized framework for conflict analysis.
\newblock In Hans~Kleine B{\"{u}}ning and Xishun Zhao, editors, {\em Theory and
  Applications of Satisfiability Testing - {SAT} 2008, 11th International
  Conference, {SAT} 2008, Guangzhou, China, May 12-15, 2008. Proceedings},
  volume 4996 of {\em Lecture Notes in Computer Science}, pages 21--27.
  Springer, 2008.
\newblock \href {https://doi.org/10.1007/978-3-540-79719-7\_3}
  {\path{doi:10.1007/978-3-540-79719-7\_3}}.

\bibitem[AFT11]{AFT11ClauseLearning}
Albert Atserias, Johannes~Klaus Fichte, and Marc Thurley.
\newblock Clause-learning algorithms with many restarts and bounded-width
  resolution.
\newblock {\em Journal of Artificial Intelligence Research},
  40:353\nobreakdash--373, January 2011.
\newblock Preliminary version in \emph{SAT~'09}.

\bibitem[And68]{Andrews68Resolution}
Peter~B. Andrews.
\newblock Resolution with merging.
\newblock {\em J. {ACM}}, 15(3):367--381, 1968.

\bibitem[BB21]{BB21QBF}
Olaf Beyersdorff and Benjamin B{\"{o}}hm.
\newblock Understanding the relative strength of {QBF} {CDCL} solvers and {QBF}
  resolution.
\newblock In James~R. Lee, editor, {\em 12th Innovations in Theoretical
  Computer Science Conference, {ITCS} 2021, January 6-8, 2021, Virtual
  Conference}, volume 185 of {\em LIPIcs}, pages 12:1--12:20. Schloss Dagstuhl
  - Leibniz-Zentrum f{\"{u}}r Informatik, 2021.
\newblock \href {https://doi.org/10.4230/LIPIcs.ITCS.2021.12}
  {\path{doi:10.4230/LIPIcs.ITCS.2021.12}}.

\bibitem[BF97]{blum1997fast}
Avrim~L Blum and Merrick~L Furst.
\newblock Fast planning through planning graph analysis.
\newblock {\em Artificial intelligence}, 90(1-2):281--300, 1997.

\bibitem[BHJ08]{BHJ08ResolutionTrees}
Samuel~R. Buss, Jan Hoffmann, and Jan Johannsen.
\newblock Resolution trees with lemmas: Resolution refinements that
  characterize {DLL}-algorithms with clause learning.
\newblock {\em Logical Methods in Computer Science}, 4(4:13), December 2008.

\bibitem[BKS04]{BKS04TowardsUnderstanding}
Paul Beame, Henry Kautz, and Ashish Sabharwal.
\newblock Towards understanding and harnessing the potential of clause
  learning.
\newblock {\em Journal of Artificial Intelligence Research},
  22:319\nobreakdash--351, December 2004.
\newblock Preliminary version in \emph{IJCAI~'03}.

\bibitem[BN21]{BN21Proof}
Sam Buss and Jakob Nordström.
\newblock Proof complexity and {SAT} solving.
\newblock In Armin Biere, Marijn Heule, Hans van Maaren, and Toby Walsh,
  editors, {\em Handbook of Satisfiability}, volume 336 of {\em Frontiers in
  Artificial Intelligence and Applications}, chapter~7, pages 233--350. {IOS}
  Press, 2nd edition, 2021.
\newblock \href {https://doi.org/10.3233/FAIA200990}
  {\path{doi:10.3233/FAIA200990}}.

\bibitem[BS97]{BS97UsingCSP}
Roberto~J. {Bayardo~Jr.} and Robert Schrag.
\newblock Using {CSP} look-back techniques to solve real-world {SAT} instances.
\newblock In {\em Proceedings of the 14th National Conference on Artificial
  Intelligence ({AAAI~'97})}, pages 203\nobreakdash--208, July 1997.

\bibitem[CGP{\etalchar{+}}08]{cadar2008exe}
Cristian Cadar, Vijay Ganesh, Peter~M Pawlowski, David~L Dill, and Dawson~R
  Engler.
\newblock {EXE: Automatically Generating Inputs of Death}.
\newblock {\em ACM Transactions on Information and System Security (TISSEC)},
  12(2):1--38, 2008.

\bibitem[DHN07]{DHN07Towards}
Nachum Dershowitz, Ziyad Hanna, and Alexander Nadel.
\newblock Towards a better understanding of the functionality of a
  conflict-driven {SAT} solver.
\newblock In Jo{\~{a}}o Marques{-}Silva and Karem~A. Sakallah, editors, {\em
  Theory and Applications of Satisfiability Testing - {SAT} 2007, 10th
  International Conference, Lisbon, Portugal, May 28-31, 2007, Proceedings},
  volume 4501 of {\em Lecture Notes in Computer Science}, pages 287--293.
  Springer, 2007.
\newblock \href {https://doi.org/10.1007/978-3-540-72788-0\_27}
  {\path{doi:10.1007/978-3-540-72788-0\_27}}.

\bibitem[dV94]{delVal94Tractable}
Alvaro del Val.
\newblock Tractable databases: How to make propositional unit resolution
  complete through compilation.
\newblock In Jon Doyle, Erik Sandewall, and Pietro Torasso, editors, {\em
  Proceedings of the 4th International Conference on Principles of Knowledge
  Representation and Reasoning (KR'94). Bonn, Germany, May 24-27, 1994}, pages
  551--561. Morgan Kaufmann, 1994.

\bibitem[DVT07]{dolby2007security}
Julian Dolby, Mandana Vaziri, and Frank Tip.
\newblock {Finding Bugs Efficiently With a {SAT} Solver}.
\newblock In {\em Proceedings of the 6th joint meeting of the European Software
  Engineering Conference and the {ACM} {SIGSOFT} International Symposium on
  Foundations of Software Engineering}, pages 195--204, 2007.
\newblock \href {https://doi.org/10.1145/1287624.1287653}
  {\path{doi:10.1145/1287624.1287653}}.

\bibitem[FB20]{FB20ClauseSize}
Nick Feng and Fahiem Bacchus.
\newblock Clause size reduction with all-uip learning.
\newblock In Luca Pulina and Martina Seidl, editors, {\em Theory and
  Applications of Satisfiability Testing - {SAT} 2020 - 23rd International
  Conference, Alghero, Italy, July 3-10, 2020, Proceedings}, volume 12178 of
  {\em Lecture Notes in Computer Science}, pages 28--45. Springer, 2020.
\newblock \href {https://doi.org/10.1007/978-3-030-51825-7\_3}
  {\path{doi:10.1007/978-3-030-51825-7\_3}}.

\bibitem[HBPV08]{HBPV08ClauseLearning}
Philipp Hertel, Fahiem Bacchus, Toniann Pitassi, and Allen {Van Gelder}.
\newblock Clause learning can effectively {P}-simulate general propositional
  resolution.
\newblock In {\em Proceedings of the 23rd National Conference on Artificial
  Intelligence ({AAAI}~'08)}, pages 283\nobreakdash--290, July 2008.

\bibitem[LFV{\etalchar{+}}20]{li2020towards}
Chunxiao Li, Noah Fleming, Marc Vinyals, Toniann Pitassi, and Vijay Ganesh.
\newblock Towards a complexity-theoretic understanding of restarts in sat
  solvers.
\newblock In {\em Theory and Applications of Satisfiability Testing--SAT 2020:
  23rd International Conference, Alghero, Italy, July 3--10, 2020, Proceedings
  23}, pages 233--249. Springer, 2020.

\bibitem[MLM21]{HandbookCDCL}
Jo{\~{a}}o Marques{-}Silva, In{\^{e}}s Lynce, and Sharad Malik.
\newblock Conflict-driven clause learning {SAT} solvers.
\newblock In Armin Biere, Marijn Heule, Hans van Maaren, and Toby Walsh,
  editors, {\em Handbook of Satisfiability - Second Edition}, volume 336 of
  {\em Frontiers in Artificial Intelligence and Applications}, pages 133--182.
  {IOS} Press, 2021.
\newblock \href {https://doi.org/10.3233/FAIA200987}
  {\path{doi:10.3233/FAIA200987}}.

\bibitem[MMZ{\etalchar{+}}01]{MMZZM01Engineering}
Matthew~W. Moskewicz, Conor~F. Madigan, Ying Zhao, Lintao Zhang, and Sharad
  Malik.
\newblock Chaff: {E}ngineering an efficient {SAT} solver.
\newblock In {\em Proceedings of the 38th Design Automation Conference
  (DAC~'01)}, pages 530\nobreakdash--535, June 2001.

\bibitem[MPR20]{MPR20CDCL}
Nathan Mull, Shuo Pang, and Alexander~A. Razborov.
\newblock On {CDCL}-based proof systems with the ordered decision strategy.
\newblock In {\em Proceedings of the 23rd International Conference on Theory
  and Applications of Satisfiability Testing ({SAT}~'20)}, volume 12178 of {\em
  Lecture Notes in Computer Science}, pages 149\nobreakdash--165. Springer,
  July 2020.

\bibitem[MS85]{MS85Solving}
Burkhard Monien and Ewald Speckenmeyer.
\newblock Solving satisfiability in less than $2^n$ steps.
\newblock {\em Discret. Appl. Math.}, 10(3):287--295, 1985.
\newblock \href {https://doi.org/10.1016/0166-218X(85)90050-2}
  {\path{doi:10.1016/0166-218X(85)90050-2}}.

\bibitem[MS99]{MS99Grasp}
Jo{\~a}o~P. {Marques-Silva} and Karem~A. Sakallah.
\newblock {GRASP}: A search algorithm for propositional satisfiability.
\newblock {\em IEEE Transactions on Computers}, 48(5):506\nobreakdash--521, May
  1999.
\newblock Preliminary version in \emph{ICCAD~'96}.

\bibitem[PD08]{PD08Learning}
Knot Pipatsrisawat and Adnan Darwiche.
\newblock A new clause learning scheme for efficient unsatisfiability proofs.
\newblock In Dieter Fox and Carla~P. Gomes, editors, {\em Proceedings of the
  Twenty-Third {AAAI} Conference on Artificial Intelligence, {AAAI} 2008,
  Chicago, Illinois, USA, July 13-17, 2008}, pages 1481--1484. {AAAI} Press,
  2008.
\newblock URL: \url{http://www.aaai.org/Library/AAAI/2008/aaai08-243.php}.

\bibitem[PD11]{PD11OnThePower}
Knot Pipatsrisawat and Adnan Darwiche.
\newblock On the power of clause-learning {SAT} solvers as resolution engines.
\newblock {\em Artificial Intelligence}, 175(2):512\nobreakdash--525, February
  2011.
\newblock Preliminary version in \emph{CP~'09}.

\bibitem[Rya04]{Ryan04Thesis}
Lawrence Ryan.
\newblock Efficient algorithms for clause-learning {SAT} solvers.
\newblock Master's thesis, Simon Fraser University, 2004.

\bibitem[SB09]{SB09Minimizing}
Niklas Sörensson and Armin Biere.
\newblock Minimizing learned clauses.
\newblock In {\em Proceedings of the 12th International Conference on Theory
  and Applications of Satisfiability Testing ({SAT}~'09)}, volume 5584 of {\em
  Lecture Notes in Computer Science}, pages 237\nobreakdash--243. Springer,
  July 2009.

\bibitem[{Van}05]{VanGelder05PoolResolution}
Allen {Van Gelder}.
\newblock Pool resolution and its relation to regular resolution and {DPLL}
  with clause learning.
\newblock In {\em Proceedings of the 12th International Conference on Logic for
  Programming, Artificial Intelligence, and Reasoning ({LPAR}~'05)}, volume
  3835 of {\em Lecture Notes in Computer Science}, pages 580\nobreakdash--594.
  Springer, 2005.

\bibitem[Vin20]{Vinyals20HardExamples}
Marc Vinyals.
\newblock Hard examples for common variable decision heuristics.
\newblock In {\em Proceedings of the 34th {AAAI} Conference on Artificial
  Intelligence ({AAAI}~'20)}, pages 1652\nobreakdash--1659, February 2020.

\bibitem[XA05]{xie2005security}
Yichen Xie and Alexander Aiken.
\newblock {Saturn: A SAT-Based Tool for Bug Detection}.
\newblock In {\em Proceedings of the 17th International Conference on Computer
  Aided Verification, {CAV} 2005}, pages 139--143, 2005.
\newblock \href {https://doi.org/10.1007/11513988\_13}
  {\path{doi:10.1007/11513988\_13}}.

\bibitem[ZMMM01]{ZMMM01EfficientConflict}
Lintao Zhang, Conor~F. Madigan, Matthew~W. Moskewicz, and Sharad Malik.
\newblock Efficient conflict driven learning in {B}oolean satisfiability
  solver.
\newblock In {\em Proceedings of the IEEE/ACM International Conference on
  Computer-Aided Design ({ICCAD}~'01)}, pages 279\nobreakdash--285, November
  2001.

\bibitem[ZMW{\etalchar{+}}18]{ZMWLCG18Effect}
Edward Zulkoski, Ruben Martins, Christoph~M. Wintersteiger, Jia~Hui Liang,
  Krzysztof Czarnecki, and Vijay Ganesh.
\newblock The effect of structural measures and merges on {SAT} solver
  performance.
\newblock In John~N. Hooker, editor, {\em Principles and Practice of Constraint
  Programming - 24th International Conference, {CP} 2018, Lille, France, August
  27-31, 2018, Proceedings}, volume 11008 of {\em Lecture Notes in Computer
  Science}, pages 436--452. Springer, 2018.
\newblock \href {https://doi.org/10.1007/978-3-319-98334-9\_29}
  {\path{doi:10.1007/978-3-319-98334-9\_29}}.

\end{thebibliography}
\bibliographystyle{alphaurl}

\ifcsname OT1/cantarell-TLF/m/n/sub\endcsname
\error There was a bad font substitution (cantarell-TLF). Try fca.
\fi
\ifcsname OT1/fca/m/n/sub\endcsname
\error There was a bad font substitution (fca). Try cantarell-TLF,
\fi

\end{document}